\documentclass[11pt,a4paper]{article}
\usepackage{macros}
\usepackage{amsfonts
}
\usepackage{amsmath,float}
\usepackage{thmtools}
\usepackage{amssymb}
\usepackage{makecell}
\usepackage{braket}
\usepackage[utf8]{inputenc}
\usepackage{hyperref} 
\usepackage{ytableau, comment}
\usepackage{tikz}
\usepackage{rotating}
\usepackage{mleftright}
\mleftright

\newtheorem{example}{Example}
\renewcommand\thmcontinues[1]{Continued}
\newenvironment{proof}{\paragraph{Proof:}}{\hfill$\square$}

\newcommand{\EQ}[1]{\begin{align}\begin{split} #1
\end{split}\end{align}}
\newcommand{\eq}[1]{\begin{equation}{ #1 
}\end{equation}}

\newcommand\blfootnote[1]{%
\begingroup 
\renewcommand\thefootnote{}\footnote{#1}%
\addtocounter{footnote}{-1}%
\endgroup 
}

\begin{document}

\title{Counting Bethe States in Twisted Spin Chains}

\author[1]{Hongfei Shu,}
\author[2]{Peng Zhao,}
\author[3,4]{Rui-Dong Zhu,}
\author[5,6]{Hao Zou\blfootnote{*The authors are ordered alphabetically and should all be viewed as co-first authors. }}
\affiliation[1]{Institute for Astrophysics, School of Physics,
Zhengzhou University, Zhengzhou,\\ Henan 450001, China}
\affiliation[2]{Joint School of the National University of Singapore and Tianjin University,\\
International Campus of Tianjin University, Fuzhou, 350207, China}
\affiliation[3]{Institute for Advanced Study \& School of Physical Science and Technology,\\ Soochow University, Suzhou 215006, China}
\affiliation[4]{Jiangsu Key Laboratory of Frontier Material Physics and Devices,\\ Soochow University, Suzhou 215006, China}
\affiliation[5]{Center for Mathematics and Interdisciplinary Sciences, Fudan University, Shanghai 200433, China}
\affiliation[6]{Shanghai Institute for Mathematics and Interdisciplinary Sciences, Shanghai 200433, China}

\emailAdd{shu@zzu.edu.cn, pzhao@tjufz.org.cn, rdzhu@suda.edu.cn, hzou@simis.cn}
\abstract{
We present a counting formula that relates the number of physical Bethe states of integrable models with a twisted boundary condition to the number of states in the untwisted or partially twisted limit.}

\allowdisplaybreaks

\maketitle


\section{Introduction}

One of the oldest problems in quantum many-body systems is the composition of many spins, which attracts continued interest to this day (see \cite{Kulish_2012, Curtright:2016eni, Gyamfi_2018, postnova2020multiplicities, Polychronakos:2023yhq} for some recent works in the mathematics and physics literature). The number of configurations of $L$ electrons with $M$ down spins is given by ${L \choose M}$. The number of highest-weight states of $SU(2)$ with a total $z$ spin of $S^z = L/2 - M$ is given by the \emph{difference} 
\eq{
\mu_{(L-M, M)} = {L \choose M} - {L \choose M-1} \,.
\label{dc}
}
This formula, first computed by Heitler \cite{Heitler}, played an important role in Heisenberg's model of a one-dimensional chain of electrons with the exchange interaction \cite{Heisenberg}. It turns out that an individual spin state is not an eigenstate of the nearest-neighbor Hamiltonian. Bethe's celebrated solution of this problem postulates the eigenstates in terms of a superposition of state vectors whose partial waves satisfy a certain quantization condition, the so-called Bethe-ansatz equations (BAEs) \cite{Bethe:1931hc}. Physical solutions parametrize the highest-weight states. The complete $2^L$-dimensional Hilbert space is spanned by the physical Bethe states and their descendants obtained by the action of the spin-lowering operator.  In mathematical terms, the number of physical Bethe states is the multiplicity of the decomposition of the $L$-fold tensor power of the fundamental representation of $\mathfrak{su}(2)$ into irreducible representations. The term subtracted in \eqref{dc} counts the descendants in the highest-weight module. The form depends only on the symmetry and not on the specific representation. For example, the formula for higher spins takes the same form as $s=1/2$ \cite[Corollary~3]{Kirillov1985}
\eq{
\mu_{s,L}(M) = c_{s,L}(M) - c_{s,L}(M-1) \,,
}
where $c_{s,L}(M)$ generalizes the binomial coefficient.

While most works on spin chains have focused on the spin-1/2 representation of $\mathfrak{su}(2)$ due to the close links to experiments, 
the Bethe ansatz was extended to multicomponent models with extended symmetries by Yang and Sutherland \cite{Yang:1967ue, Sutherland:1975vr} and in models of higher spin by Faddeev and collaborators (\cite{Faddeev:1996iy} provides a standard review). One of the great achievements of the algebraic Bethe ansatz was the construction of an integrable model with a higher spin. 

The global symmetry may be broken by turning on a transverse magnetic flux through the periodic chain, which is equivalent to a twisted boundary condition \cite{Byers:1961zz}. The states will no longer be organized into the highest-weight modules of the original symmetry group. There are several advantages in considering the twisted model. First, not all solutions to the BAEs correspond to eigenstates of the Hamiltonian. These unphysical states must be carefully removed by imposing additional conditions. It is empirically found that such unphysical solutions disappear in the twisted model and all solutions are physical. It is conjectured from numerical studies that unphysical solutions do not appear in twisted spin chains (cf. \cite{Bazhanov:2010ts, Hou:2023ndn}). In higher rank, one may consider partially twisting the model. The symmetry is broken to a subgroup and we expect a similar difference formula to hold for the branching coefficients. It is somewhat surprising that a formula generalizing \eqref{dc} to higher rank has not appeared explicitly in the literature (see \cite{Kirillov1985, Kirillov_1987} for a different combinatorial approach to the problem). 

In this paper, we derive a formula that expresses the tensor-product multiplicities of the composition of $L$ spins in the $2s$-symmetric representation of $SU(r+1)$ as an alternating sum of the number of spin configurations. It connects the number of physical solutions in the untwisted spin chain with the number of solutions in the twisted chain. Our formula follows from the idea that some of the physical states in the twisted case become descendants of the highest-weight states. The completeness of the Hilbert space is guaranteed, since the physical solutions that span the Hilbert space recombine into the highest-weight modules in the untwisted limit. Our results are purely obtained from the representation theory. 

Our work is motivated in part by high-energy physics. The Bethe/Gauge correspondence connects the Bethe states with the vacua of supersymmetric gauge theories \cite{Nekrasov:2009uh,Nekrasov:2009ui}. In the dual gauge theory, turning on the twisted boundary condition corresponds to turning on a non-zero Fayet-Iliopoulos parameter and naturally avoiding singularities. Yet not much is known about the properties on the phase boundaries. The number of physical solutions of the twisted spin chain is the number of supersymmetric vacua, and it can be counted by the Witten index \cite{Shu:2022vpk}, which is also interpreted as solving a 3d restricted-occupancy problem as explained in Section~\ref{s:counting}. However, in the untwisted limit when the gauge theory approaches the singularities of its moduli space, the Witten index will not be well defined anymore \cite{Witten:1993yc}. If we extend the Bethe/Gauge correspondence and trust that the untwisted Bethe equations still describe the vacua in the gauge theory, then we learn about properties of the singularity that are inaccessible using traditional methods in quantum field theory. 

This paper is organized as follows. In Section \ref{s:ABA}, we review the algebraic Bethe ansatz and analyze the symmetries preserved by partial twists. In Section \ref{s:counting}, we first formulate the counting of spin configurations combinatorially as a restricted-occupancy problem. We then present the counting formula that relates the tensor-product multiplicities to the restricted-occupancy coefficients. The completeness of the Hilbert space will also be discussed at the end of this section. In Section \ref{s:example},  we study partial twists and present the counting formula for the branching coefficients. In Section \ref{s:gen}, we study generalizations to Kondo-type models and models with Lie superalgebras. 

\section{Algebraic Bethe ansatz with a twist}\label{s:ABA}

\subsection{XXX spin chain}
The Hamiltonian of the XXX spin chain is embedded in one of the conserved charges generated by the transfer matrix, constructed as follows. First, introduce an auxiliary space $V_0 = \mathbb{C}^2$ in addition to the physical space $V_n = \mathbb{C}^{2s+1}$ ($n=1, \ldots, L$) of the spins.
The $R$ matrix $R_{00'}(u)\in{\rm End}(\mathbb{C}^2\otimes \mathbb{C}^2)$ and the Lax matrix $\mathcal{L}_{0n}(u)\in{\rm End}(\mathbb{C}^2\otimes V_n)$ are given by 
\eq{
    R_{00'}(u)=\lt(\begin{array}{cccc}
     u+i & 0 & 0 & 0\\
    0 &  u & i & 0\\
    0 & i &  u & 0\\
    0 & 0 & 0 &  u+i\\
    \end{array}\rt), \quad \mathcal{L}_{0n}(u)=\lt(\begin{array}{cc}
    \lt( u+\frac{i}{2}\rt)\mathbb{I}+iJ_n^z & iJ_n^-\\
    iJ_n^+ & \lt( u+\frac{i}{2}\rt)\mathbb{I}-iJ_n^z\\
    \end{array}
    \rt)\,,
}
with the $\mathfrak{su}(2)$ generators satisfying 
\EQ{
    \lt[J_n^+,J_n^-\rt]=2J_n^z\,,\quad \lt[J_n^z,J_n^\pm\rt]=\pm J_n^\pm\,,
}
taken in the spin-$s$ representation. They satisfy the celebrated RLL relation:
\eq{
    R_{00'}( u- u')\mathcal{L}_{0n}(u)\mathcal{L}_{0'n}(u')=\mathcal{L}_{0'n}(u')\mathcal{L}_{0n}(u)R_{00'}(u- u')\,.
\label{RLL}
}
We define the monodromy matrix ${\bf T}( u)\in {\rm End}(\bigotimes_{i=0}^LV_i)$ as a product of the Lax matrices
\EQ{
    {\bf T}( u)&:= \mathcal{L}_{0L}( u)\cdots \mathcal{L}_{02}( u)\mathcal{L}_{01}( u)\\
    &=\lt(\begin{array}{cc}
    {\cal A}( u) & {\cal B}( u)\\
    {\cal C}( u) & {\cal D}( u)\\
    \end{array}
    \rt) \,,
\label{def-transfer}
}
where ${\cal A},{\cal B}, {\cal C}, {\cal D}\in {\rm End}(\bigotimes_{i=1}^LV_i)$ are operators defined on the physical space. The RLL relation implies an analogous RTT relation on the monodromy matrices:
\eq{
 R_{00'}( u- u'){\bf T}_{0}(u){\bf T}_{0'}(u')={\bf T}_{0'}(u'){\bf T}_{0}(u)R_{00'}(u- u')\,.
 \label{RTT}
}
By taking the trace over the auxiliary spaces labeled by $0$ and $0'$, we may show that the transfer matrix, defined as 
\eq{
t( u):= \rm{tr}_0 {\bf T}( u)\,,
} 
commutes at different values of spectral parameters $u$.

To diagonalize $t(u)$, one starts with the ground state satisfying the condition
\eq{
    {\cal C}( u)\ket{\Omega}=0 \,,
}
for all $ u\in\mathbb{C}$ and constructs the eigenstates, known as the Bethe states: 
\eq{
\ket{\Psi(\vec  u)} :=  {\cal B}( u_1){\cal B}( u_2)\cdots{\cal B}( u_M)\ket{\Omega}\,.\label{ABA-A1}
}
The quantization condition on the spectral parameters $ u_i$, known as the Bethe-ansatz equations, ensures that the Bethe states are eigenstates. We define the generators of a global $SU(2)$ symmetry as 
\eq{
    S^\pm:= \sum_{n=1}^LJ_n^{\pm}\,,\quad S^z:= \sum_{n=1}^LJ^{z}_n\,.
}
The following statements will be essential for our discussion:
\begin{enumerate}
    \item The transfer matrix is invariant under the global $SU(2)$ symmetry.
    \item The Bethe states are the highest-weight states of the global $SU(2)$ symmetry.
\end{enumerate}

Now we review the technical details to establish the above statements. One may read off from the RTT relation \eqref{RTT} that
\EQ{
&( u_1- u_2+i){\cal B}( u_1){\cal A}( u_2)=( u_1- u_2){\cal A}( u_2){\cal B}( u_1)+i{\cal B}( u_2){\cal A}( u_1)\,,
\\
    &( u_1- u_2+i){\cal B}( u_2){\cal D}( u_1)=( u_1- u_2){\cal D}( u_1){\cal B}( u_2)+i{\cal B}( u_1){\cal D}( u_2)\,,
\\
    &( u_1- u_2+i){\cal C}( u_2){\cal A}( u_1)=( u_1- u_2){\cal A}( u_1){\cal C}( u_2)+i{\cal C}( u_1){\cal A}( u_2)\,,
\\
    &( u_1- u_2+i){\cal C}( u_1){\cal D}( u_2)=( u_1- u_2){\cal D}( u_2){\cal C}( u_1)+i{\cal C}( u_2){\cal D}( u_1)\,.\label{ABCD}
}
One may show from the large-$ u$ expansion of \eqref{def-transfer} that
\eq{
    {\cal A}( u)\sim  u^L\bigotimes_{i=1}^L\mathbb{I}_{i}\,,\quad 
    {\cal B}( u)\sim i u^{L-1}S^-\,,\quad 
    {\cal C}( u)\sim i u^{L-1}S^+\,, \quad 
    {\cal D}( u)\sim  u^L\bigotimes_{i=1}^L\mathbb{I}_{i}\,.
\label{large}
}
By taking the large-$ u$ limit of one of the spectral parameters in \eqref{ABCD}, we obtain 
\EQ{
    &\lt[{\cal A}( u),S^+\rt]= -\lt[{\cal D}( u),S^+\rt]={\cal C}( u)\,, \\
    &\lt[{\cal A}( u),S^-\rt]=-\lt[{\cal D}( u),S^-\rt]=-{\cal B}( u)\,.
}
By further using the identity 
\eq{
    ( u_1- u_2)\lt[{\cal B}( u_1),{\cal C}( u_2)\rt]=i\lt({\cal D}( u_2){\cal A}( u_1)-{\cal D}( u_1){\cal A}( u_2)\rt)\,,
}
which follows again from the RTT relation, we obtain 
\eq{
    \lt[{\cal B}( u),S^+\rt]=-\lt[{\cal C}( u),S^-\rt]={\cal D}( u)-{\cal A}( u)\,.
    \label{BC}
}
One can then compute the commutator between $S^z$ and ${\cal A}$, ${\cal D}$ with the Jacobi identity to obtain 
\eq{
    \lt[{\cal A}( u),S^z\rt]=\frac{1}{2}\lt[{\cal A}( u),\lt[S^+,S^-\rt]\rt]=0\,,\quad \lt[{\cal D}( u),S^z\rt]=\frac{1}{2}\lt[{\cal D}( u),\lt[S^+,S^-\rt]\rt]=0\,.
}
Thus the transfer matrix $t( u) = {\cal A}( u) + {\cal D}( u)$ commutes with all the $\mathfrak{su}(2)$ generators.

To demonstrate that the Bethe states are eigenstates of the transfer matrix, we commute ${\cal A}$ and ${\cal D}$ past ${\cal B}$ using \eqref{ABCD} and require the vanishing of the extra terms. Finally, one may show that the Bethe states are highest-weight states, that is $S^+ \ket{\Psi(\vec  u)}=0$, by applying \eqref{BC} and removing the extra terms using the BAEs. By \eqref{large}, a descendant may be obtained from a Bethe state by sending one or more rapidity parameters to infinity.

To include a twisted boundary condition, we insert a phase operator into the transfer matrix 
\eq{
    t_{\theta}( u):= \rm{tr}_0\Big[{\rm Ph}(\theta){\bf T}( u)\Big]\,,
}
where ${\rm Ph}(\theta)\in{\rm End}(\mathbb{C}^2)$ defined on the auxiliary space $V_0=\mathbb{C}^2$ is
\eq{
    {\rm Ph}(\theta)=\lt(\begin{array}{cc}
    e^{i\theta/2} & 0\\
    0 & e^{-i\theta/2}\\
    \end{array}\rt)\,.
}
One of the remarkable achievements of the algebraic Bethe ansatz is the generalization to an arbitrary spin $s$. Although the Hamiltonian takes a complicated form, the BAEs is a straightforward generalization of the $s=1/2$ case. We may consider the spin-$s_i$ representation of $\mathfrak{su}(2)$ at the $i$-th site of the spin chain, and it is equivalent to taking the $(2s_i+1)$-dimensional representation of $J^{\pm}_i$ and $J_{i}^z$ and to set $V_i=\mathbb{C}^{2s_i+1}$. 
It is then easy to follow the algebraic Bethe ansatz to see that the representation-theoretic information of the spin chain is only contained in the actions of ${\cal A}$ and ${\cal D}$ on the ground state. 
\EQ{
    {\cal A}( u)\ket{\Omega}=\prod_{l=1}^L\lt(\lt( u+\frac{i}{2}\rt)+iJ_{l}^z\rt)\ket{\Omega},\\
    {\cal D}( u)\ket{\Omega}=\prod_{l=1}^L\lt(\lt( u+\frac{i}{2}\rt)-iJ_{l}^z\rt)\ket{\Omega}.
}
The BAEs with a different spin at each site and a twist angle can then be found as 
\eq{
    e^{i\theta}\prod_{n=1}^L\frac{v_j+is_n}{v_j-is_n}=-\prod_{k=1}^M\frac{v_j-v_k+i}{v_j-v_k-i}\,,\label{BAE-A1}
}
where we put $v_j= u_j+\frac{i}{2}$. In this paper, we focus on the case where all the $s_n$ are equal to $s$, but the counting formula does not depend on this assumption and can be further generalized to spin chains with an inhomogeneity at each site. 

\subsection{Singular solutions and physical conditions}
While the number of physical Bethe states is determined by symmetry, not all solutions to the BAEs correspond to the physical states. There are solutions that either give divergent eigenvalues or zero eigenstates. For example, when $s=1/2$, the energy eigenvalue is given by \cite{Faddeev:1996iy}
\eq{\label{eq:BAE-spe}
    E =\sum_{j=1}^M\frac{1}{v^2_j+\frac{1}{4}}-L\,.
}
The Bethe roots at $v_j=\pm i/2$ give divergent eigenvalues, and solutions containing such Bethe roots generically cannot be physical. A cure to regularize such singular solutions is to turn on a twist angle.
If a pair of Bethe roots, say $v_1$ and $v_2$, take the form \cite{Nepomechie:2013mua}
\eq{
    v_1=\frac{i}{2}+\epsilon+c_1\epsilon^L+{\cal O}(\epsilon^{L+1})\,,\quad v_2=-\frac{i}{2}+\epsilon+c_2\epsilon^L+{\cal O}(\epsilon^{L+1})\,,\label{root-reg}
}
with some regularization parameter $\epsilon$ that goes to zero when $\theta\rightarrow 0$, then a finite value of the energy can be obtained. The physical condition for such singular solutions can be derived by substituting \eqref{root-reg} into the BAEs and eliminating $c_1-c_2$, 
\EQ{
    &e^{2i\theta}(-1)^L\prod_{j=3}^{M}\frac{(v_j+\frac{i}{2})(v_j+\frac{3i}{2})}{(v_j-\frac{i}{2})(v_j-\frac{3i}{2})}=1\,,\\
    &e^{i\theta}\frac{(v_j+\frac{i}{2})^{L-1}}{(v_j+\frac{i}{2})^{L-1}}\frac{(v_j-\frac{3i}{2})}{(v_j+\frac{3i}{2})}\prod_{k=3,k\neq j}^M\frac{v_j-v_k-i}{v_j-v_k+i}=1\,,\quad j=3,\cdots,M\,.
}
For a spin-$s$ chain, similar physical conditions for singular solutions may be derived \cite{Nepomechie:2014hma}. 

\subsection{Higher-rank spin chains and nested Bethe ansatz}\label{s:rank2-NBA}

We briefly describe how the algebraic Bethe ansatz is modified in higher rank \cite{Kulish:1979cr, Kulish:1983rd}. The following $R$ matrix of dimension $(r+1)^2$ defined by 
\eq{
    R( u)= u\mathbb{I}+i\mathbb{P}\,,
}
satisfies the Yang-Baxter equation, where 
\eq{
    \mathbb{P}=\sum_{i,j=1}^{r+1} E_{ij}\otimes E_{ji}\,,\quad (E_{ij})_{kl}=\delta_{i,k}\delta_{j,l}\,,
}
plays the role of the permutation operator on the vector space $\mathbb{C}^{r+1}\otimes \mathbb{C}^{r+1}$.  The Lax matrix may be written as
\eq{\mathcal{L}( u) =  u + i \sum_{i,j =1}^{r+1} E_{ij} X^{ji}\,,}
where $X^{ji}$ are the generators of $\mathfrak{sl}(r+1)$, which can be taken to any representation to construct different spin chains. As in the rank-1 case,  the monodromy matrix, defined as \eqref{def-transfer}, 
satisfies the RTT relations \eqref{RTT}.  From the RTT relation, we can find the commutation relations satisfied by the operators in the components of ${\bf T}$.
Now ${\cal D}$ is an $r\times r$ block, while ${\cal B}$ (${\cal C}$) a row (column) with $r$ components. The commutation relations are almost the same as \eqref{ABCD}, with ${\cal B}, {\cal C}$ and ${\cal D}$ replaced by ${\cal B}_i, {\cal C}_i$ and ${\cal D}_{ij}$, respectively.  We then introduce the twisted transfer matrix
\eq{
    t_{\theta}( u):= {\rm tr}_0\,\Big[{\rm Ph}(\vec{\theta}){\bf T}( u)\Big]\,,
}
with 
\eq{
    {\rm Ph}(\vec \theta)={\rm diag}(e^{i\phi_{1}},e^{i\phi_{2}}, \cdots, e^{i\phi_{r+1}})\,, \qquad \theta_i = \phi_{i} - \phi_{i+1}\,,
}    
where $\sum^{r+1}_{i=1} \phi_i = 0$.

The main task of the algebraic Bethe ansatz is to solve the eigenvalue problem
\begin{equation}\label{eq:eigen}
    t_\theta(u)\ket{\Psi}=\left[e^{i\phi_1}{\cal A}(u)+\sum_{i=2}^{r+1}\,e^{i\phi_i}{\cal D}_{ii}(u)\right]\ket{\Psi}=\Lambda\ket{\Psi},
\end{equation}
which can be achieved by the following recursive process\footnote{We would like to express our gratitude to Prof. Wenli Yang for helping to clarify our confusion regarding the nested Bethe ansatz.}.
First we focus on the case $r=2$, where we introduce the notations 
\begin{equation}
    {\cal B}_1:=(t_\theta)_{12},\quad {\cal B}_2:=(t_\theta)_{13},\quad {\cal B}^{(2)}:=(t_\theta)_{23},\quad \vec{\cal B}:=({\cal B}_1,{\cal B}_2).\label{nesting-2}
\end{equation}
We remark that ${\cal B}_1$, ${\cal B}_2$ and ${\cal B}^{(2)}$ are lowering operators in $\mathfrak{su}(3)$ with the same weight as the roots $-\alpha_1$, $-\alpha_1-\alpha_2$, $-\alpha_2$. As the initial step, the ansatz for the eigenvector is proposed as
\begin{align}
    &\ket{\Psi}=\lt[\lt(\vec{\cal B}_{1}(u_{1}^{(1)})\otimes\cdots\otimes\vec{\cal B}_{M_1}(u_{M_{1}}^{(1)})\rt)\cdot \vec{V}_{M_1}(\{u^{(2)}_i\}_{i=1}^{M_2})\rt]\ket{\Omega_3},\\
    &\vec{V}_{M_1}(\{u^{(2)}_i\}_{i=1}^{M_2}):=\prod_{i=1}^{M_2}{\cal B}^{(2)}(u^{(2)}_i)\ket{\uparrow}_{1}\otimes \dots \otimes \ket{\uparrow}_{M_1}\in \mathbb{C}^2_{1}\otimes \mathbb{C}^2_{2}\otimes \dots \otimes \mathbb{C}^2_{M_1},
\end{align}
where $M_i$ is the number of magnons, $\ket{\Omega_3}$ represents the highest-weight states in the physical Hilbert space defined by ${\cal C}_i\ket{\Omega_3}=0$, $\ket{\uparrow}=\binom{1}{0}\in \mathbb{C}^2$, $\vec{\cal B}_{i}$ denotes the operator-valued row vector $\vec{\cal B}$ acting on the $i$-th virtual vector space $\mathbb{C}^2_{i}$, and $\vec{V}_{M_{1}}$ is a column vector contracted with $\vec{\cal B}_{i}$'s. 

The above ansatz can be easily generalized to arbitrary rank $r$ of A-type Lie algebra. For example in the case of $\mathfrak{su}(4)$, note that $\ket{\Omega_3}$ for the spin chain in the fundamental representation is given by
\begin{equation}
\ket{\Omega_3}=\bigotimes_{i=1}^L\ket{1},\quad \ket{1}:=\lt(\begin{array}{c}
     1\\
     0\\
     0\\
\end{array}\rt), 
\end{equation}
it can serve as a virtual vector (by replacing $L$ by a new magnon number) to be contracted with three additional lowering operators. More explicitly, the Bethe ansatz for $\mathfrak{su}(4)$ can be written as 
\begin{align}
    \ket{\Psi}=\vec{B}^{(1)}_{1}(u^{(1)}_1)\otimes \dots\otimes \vec{B}^{(1)}_{M_1}(u^{(1)}_{M_1})\cdot\Bigg\{\vec{B}^{(2)}_{1}(u^{(2)}_1)\otimes \dots\otimes \vec{B}^{(2)}_{M_2}(u^{(2)}_{M_2})\cr
    \cdot\lt(\prod_{i=1}^{M_3}{\cal B}^{(3)}(u^{(3)}_i)\bigotimes_{j=1}^{M_2}\ket{\uparrow}_{j}\rt)\otimes \bigotimes_{k=1}^{M_1}\ket{1}_{k}\Bigg\}\ket{\Omega_4},
\end{align}
where $\ket{\Omega_4}$ is the state in the physical Hilbert space annihilated by all raising operators $(t_\theta)_{ij}$ for $i>j$, and 
\begin{equation}
    \vec{\cal B}^{(1)}:=((t_\theta)_{12},(t_\theta)_{13},(t_\theta)_{14}),\quad \vec{\cal B}^{(2)}:=((t_\theta)_{23},(t_\theta)_{24}),\quad {\cal B}^{(3)}:=(t_\theta)_{34}.
\end{equation}

Using the commutation relations of the operators, we can compute the actions of ${\cal A}$ and ${\cal D}_{ii}$ on $\ket{\Psi}$ until they act directly on the ground state $\ket{\Omega}_{r+1}$, where ${\cal A}$ and ${\cal D}_{ii}$ become diagonal. However, we will encounter some ``unwanted terms'' on the right-hand side of \eqref{eq:eigen}. To kill these terms, we should impose the first BAE. In the meantime, we should also solve the eigenvalue problem for $\sum_{i=2}^{r+1}e^{i\phi_i}{\cal D}_{ii}$, which can be achieved by repeating a similar procedure as for \eqref{eq:eigen} by splitting the operator to be diagonalized into ${\cal D}_{22}$ and $\sum_{i=3}^{r+1}e^{i\phi_i}{\cal D}_{ii}$. Repeating this procedure, we eventually obtain $r$ BAEs for $\mathfrak{su}(r+1)$ spin chain in the $2s$-th symmetric representation of the first fundamental representation (which will also be referred to as the spin-$s$ representation\footnote{In the nesting process, which realizes the embedding $\mathfrak{gl}_{2}\subseteq\mathfrak{gl}_{3}\subseteq \dots \mathfrak{gl}_{r}\subseteq \mathfrak{gl}_{r+1}$ in the monodromy matrix ${\bf T}$ (see e.g. \cite[(3.6)]{Slavnov:2019hdn}), the $2s$-th symmetric representation considered here can be viewed as the spin-$s$ representation of $\mathfrak{gl}_{2}$. This is why we alternatively refer to this representation as the spin-$s$ representation of $\mathfrak{gl}_{r+1}$ in our context. } later in this article), 
\EQ{
&-e^{i\theta_1}\frac{(v^{(1)}_j+is)^LQ_1(v^{(1)}_j-i)Q_{2}(v^{(1)}_j+i/2)}{(v^{(1)}_j-is)^LQ_1(v^{(1)}_j+i)Q_2(v^{(1)}_j-i/2)}=1\,,\\
&-e^{i\theta_2}\frac{Q_1(v^{(2)}_j+i/2)Q_2(v^{(2)}_j-i)Q_3(v^{(2)}_j+i/2)}{Q_1(v^{(2)}_j-i/2)Q_2(v^{(2)}_j+i)Q_2(v^{(2)}_j-i/2)}=1 \,,\\
&\qquad\qquad\qquad\qquad\qquad\vdots\\
&-e^{i\theta_r}\frac{Q_{r-1}(v^{(1)}_j+i/2)Q_r(v^{(1)}_j-i)}{Q_{r-1}(v^{(1)}_j-i/2)Q_r(v^{(1)}_j+i)}=1\,, 
\label{BAE-A2}
}
where $v_j^{(i)}= u_j^{(i)}+\frac{i}{2}$ and
\eq{
    Q_i(v):= \prod_{k=1}^{M_i}(v- v^{(i)}_k)\,.
}

Each Bethe vector is a highest-weight state with the Dynkin labels \cite[page L595]{Kulish:1983rd}
\eq{(2sL-2M_1+M_2, \cdots, M_{i-1} - 2M_{i} + M_{i+1}, \cdots, M_{r-1} -2M_r)
\label{weight}
\,.}
The total dimension of the Hilbert space of such an $L$-site model is given by \cite[Section~6]{Kirillov1985}
\eq{
    \dim H=\lt(\begin{array}{c}
    2s+r\\
    r\\
    \end{array}
    \rt)^L.\label{eq:dim-Hilb}
}

This concludes the brief review of the algebraic Bethe ansatz.

\subsection{Symmetries of the twisted chain}
\label{sec:symmetries}

In this paper, all the twist angles are taken modulo $2\pi \mathbb{Z}$. For the XXX spin chain, we see that after introducing a twist angle, the transfer matrix $t_\theta$ commutes only with the $S^z$ generator. The $SU(2)$ symmetry is broken to $U(1)$. In other words, when $\theta = 0$, the symmetry is enhanced and the Bethe states organize into highest-weight representations of $\mathfrak{su}(2)$. An interesting feature in higher rank is that the model may be partially twisted where only a subgroup of the non-Abelian symmetry remains. 
Therefore, the Bethe states are restricted to the irreducible representations of the associated subalgebra.

Substituting the $R$ matrix into the RTT relation \eqref{RTT}, one finds
\eq{
   \begin{aligned}
   &\Big[{\bf T}_{0}(u),{\bf T}_{0'}(u')\Big] = \frac{i}{u- u'}\Big\{{\bf T}_{0'}(u'){\bf T}_{0}(u)\mathbb{P}_{00'} 
   -\mathbb{P}_{00'}
   {\bf T}_{0}(u){\bf T}_{0'}(u')\Big\} \,.
   \end{aligned}
}
In the basis of $E_{ij}\otimes E_{kl}$, this relation becomes
\eq{
    \begin{aligned}
        &\Big[T_{ij}(u),T_{kl}(u')\Big]=\frac{i}{u- u'}\Big\{T_{kj}(u')T_{il}(u)-T_{kj}(u)T_{il}(u')\Big\} \,.
    \end{aligned}
}
We consider the large-$u'$ expansion of the monodromy matrix
\eq{
   T_{kl}(u')= u'^L
\left(\delta_{kl}+\sum_{r=1}\frac{1}{(i u')^{r}}T_{kl}^{(r)}\right)\,,
}
where $T_{kl}^{(1)}$ provides the generators of the global $SU(r+1)$ symmetry.
Taking the trace in the space $0$ on the left-hand side,
the $u'$-free term gives\footnote{A more comprehensive understanding of the symmetries of untwisted spin chains can be found in \cite{ACDFR05}, which is based on the Yangian.}
\eq{
\lt[t_\theta(u),T_{kl}^{(1)}\rt]=\big(e^{i\phi_{l}}-e^{i\phi_{k}}\big)T_{kl}(u)\,.
\label{eq:sym-trans}
}
We see that the transfer matrix commutes with the Cartan generators $T^{(1)}_{kk}$. The off-diagonal generators are in one-to-one correspondence with the root vectors. Expressing $\phi_l$ in terms of the $\theta_i$'s, we find that whenever $\theta_k + \cdots + \theta_{l-1} = 0$, the generators $T^{(1)}_{kl}$ and $T^{(1)}_{lk}$ commutes with the transfer matrix. Together with the Cartan part, they form an $\mathfrak{su}(2)$ subalgebra associated with the positive root $\alpha_k+\cdots +\alpha_{l-1}$.
It is convenient to introduce a matrix to label the correspondence between the twisted angles and the elements of the $T^{(1)}$ matrix:
\begin{equation}
\left(\begin{array}{cccccccc}
0 & \theta_{1} & \theta_{1}+\theta_{2} &  &  & \cdots & \theta_{1}+\cdots+\theta_{r-1} & \theta_{1}+\cdots+\theta_{r}\\
\theta_{1} & 0 & \theta_{2} &  &  &  & \theta_{2}+\cdots+\theta_{r-1} & \theta_{2}+\cdots+\theta_{r}\\
\theta_{1}+\theta_{2} & \theta_{2} & 0 &  &  &  & \theta_{3}+\cdots+\theta_{r-1} & \theta_{3}+\cdots+\theta_{r}\\
\\
\vdots &  &  &  &  & \ddots &  & \vdots\\
\\
\theta_{1}+\cdots+\theta_{r-1} & \theta_{2}+\cdots+\theta_{r-1} & &  &  &  & 0 & \theta_{r}\\
\theta_{1}+\cdots+\theta_{r} & \theta_{2}+\cdots+\theta_{r} & &  &  & \cdots & \theta_{r} & 0
\end{array}\right)
\,.
\end{equation}

By following a similar argument as before, we may deduce the following statements for a partially twisted chain:
\begin{enumerate}
    \item The transfer matrix is invariant under a subgroup of the $SU(r+1)$ symmetry.
    \item The Bethe states are the highest-weight states of the subgroup symmetry.
\end{enumerate}

In the fully twisted case where all the twist parameters take generic non-zero values, the symmetry is broken to $U(1)^{r}$. The symmetry is restored on each locus $\theta_k+\cdots +\theta_{l-1} = 0$.
These loci may intersect and more generators in \eqref{eq:sym-trans} will commute with the transfer matrix, leading to a larger subalgebra of $\mathfrak{su}(r+1)$. For $A_2$, if $\theta_1$, $\theta_2$ or $\theta_1+\theta_2$ vanishes, the symmetry is enhanced to $SU(2) \times U(1)$. When $\theta_1,\theta_2=0$, one recovers the $SU(3)$ symmetry.
In Table \ref{symbreak}, we present more examples of the symmetry enhancement for the $A_3$ case. On each phase boundary, an $SU(2)$ symmetry is restored. On the intersection of two such phase boundaries, there could be an $SU(3)$ or $SU(2)\times SU(2)$ symmetry, depending on whether the corresponding two positive roots can be assembled into an $\mathfrak{su}(3)$ subalgebra or not.  This can be generalized to the intersection of any number of phase boundaries. When the twist angles are turned off, {\it i.e.,} at the origin of the space of twist angles, the whole $SU(r+1)$ symmetry is recovered.

\begin{table}[!h]
\center
\begin{tabular}{c|c} \hline
Phase boundaries & Subalgebra \\
\hline
\hline
\makecell{${\theta_1= \theta_2=0}$, ${\theta_1=\theta_2+\theta_3=0}$, \\ ${\theta_2= \theta_3=0}$, ${\theta_3= \theta_1+\theta_2=0}$ }& $\mathfrak{su}(3) \oplus \mathfrak{u}(1)$ \\\hline
\makecell{${\theta_1= \theta_3=0}$, ${\theta_2= \theta_1+\theta_2+\theta_3=0}$,\\ ${\theta_1+\theta_2=\theta_2+\theta_3=0}$ }& $\mathfrak{su}(2)\oplus\mathfrak{su}(2)\oplus\mathfrak{u}(1)$\\ \hline
\makecell{${\theta_1=0}$, ${\theta_2=0}$, ${\theta_3=0}$, ${\theta_1+\theta_2=0}$, \\  ${\theta_2+\theta_3=0}$, ${\theta_1+\theta_2+\theta_3=0}$} & $\mathfrak{su}(2)\oplus \mathfrak{u}(1) \oplus \mathfrak{u}(1)$\\ \hline
\end{tabular}
\caption{Phase boundaries and their symmetries in the $SU(4)$ spin chain. The symmetry is broken to $U(1)^3$ at a generic point in the space of twist angles and is restored on the phase boundaries.} 
\label{symbreak}
\end{table}

\section{Untwisting the spin chain}\label{s:counting}

A generic twist breaks the symmetry algebra of the untwisted model to its maximal torus $\mathfrak{u}(1)^{r}$. We first decompose the Hilbert space of the spin chain into subspaces weighted by the $U(1)$'s and denote the multiplicity as $c_{s,L}(\vec{M})$. The same Hilbert space can be decomposed into subspaces of the highest-weight representations of the symmetry algebra of the periodic spin chain, and we denote the multiplicity by $\mu_{\lambda}$. $c_{s,L}(\vec{M})$ counts the  physical solutions to the BAEs \eqref{BAE-A2} with all $\theta_i\neq 0$. $\mu_{\lambda}$, on the other hand, counts the physical solutions to the BAEs in the untwisted limit with all $\theta_i=0$. In this section, we provide a formula of $\mu_{\lambda}$ in terms of the $c_{s,L}(\vec{M})$'s. 

\subsection{Spin-chain states and restricted-occupancy problem}

Counting the configurations in a spin chain may be rephrased combinatorially as a restricted-occupancy problem \cite[section~2]{Freund:1956ro}, as shown in Figure~\ref{fig: 2d}:
\begin{quote}
    {\it We place $M$ indistinguishable boxes inside a square of size $L \times 2s$, aligning to the left. The number of configurations will be denoted $c_{s,L}(M)$.
    }
\end{quote}
\begin{figure}[!h]
\centering
	\includegraphics[scale=0.20]{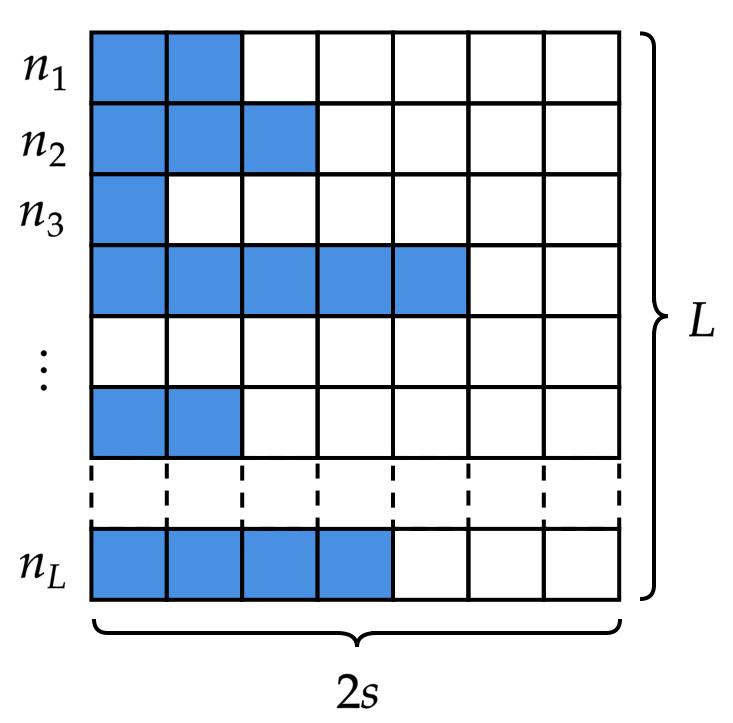}
	\caption{The Hilbert space of an $\mathfrak{su}(2)$ spin-$s$ spin chain corresponds to a restricted-occupancy problem. $n_i$ is the number of boxes in the $i$-th row and $\sum_{i=1}^L n_i = M$.}
	\label{fig: 2d}
\end{figure}
 
A powerful tool to solve enumerative problems in combinatorics is to construct a generating function whose coefficients give the answer we are looking for. For the above specific problem, the generating function has been given in \cite{Freund:1956ro}: 
\eq{
\label{eq:2dgen}
    g_{s,L}(x):=  \left(\frac{1-x^{2s+1}}{1-x}\right)^L\,.
}
The coefficient of $x^{M}$ in the expansion of $g_{s,L}(x)$ is the answer to the restricted-occupancy problem stated above.

For higher-rank spin chains, one may formulate a generalized combinatorial problem, which was called {\it 3d restricted occupancy} in \cite{Shu:2022vpk}, as shown in Figure~\ref{fig: 3d}:
\begin{quote}
    {\it Suppose now we have a total of $M=M_1 + \cdots + M_r$ indistinguishable 3d boxes. The task is again to place them inside a square of size $L\times 2s$ and align them to the left. A box can be placed on top of another with a maximum of $r$ layers. The number of boxes in the $i$-th layer should be $M_i$ with $2sL \geq M_1\geq M_2 \geq \cdots \geq M_r$.  How many configurations are there?}
\end{quote}
\begin{figure}[!h]
\centering
	\includegraphics[scale=0.5]{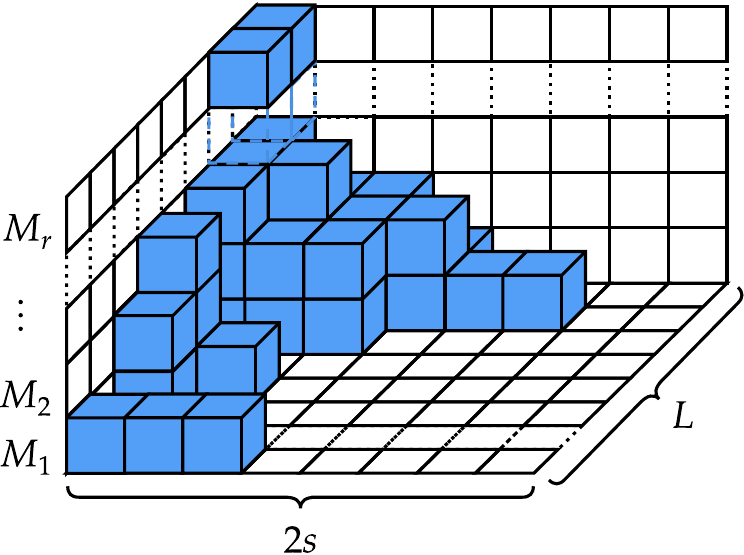}
	\caption{A 3d restricted-occupancy problem. On each level, we have a 2d restricted-occupancy with $n_\alpha^{(a)}$ boxes in each row and a total of $M_a$ boxes. Each row on a higher level has no more boxes than the lower level.}
	\label{fig: 3d}
\end{figure}

It is also tempting to construct a generating function to solve this 3d restricted-occupancy problem. First, we need to introduce $r$ formal variables denoted by $x_i$ so that in the expansion of the generating function, the coefficients of $x_1^{M_1}x_2^{M_2}\cdots x_r^{M_r}$ would give the answer. Second, the meaning of ``a box can be placed on top of another'' implies a nesting structure in the generating function. For $r=2$, a na\"ive generating function can be written as 
\eq{\label{eq:gen2layers1}
g_{s,L}(x_1, x_2) = \left[1 + x_1(1+x_2) + x_1^2(1+x_2+x_2^2) + \cdots + x_1^{2s} (1+ x_2 + \cdots + x_2^{2s})\right]^L\,.
}
However, this is too cumbersome to generalize. The trick is to take advantage of this nesting structure and encode it in the variables. One can immediately recognize that if we view the above na\"ive generating function in terms of the ``nesting variables'' $\{1, x_1, x_1 x_2\}$, $g_{s,L}(x_1, x_2)$ is just the $L$-th power of the Schur polynomial determined by a one-row Young diagram $\lambda = (2s,0,\ldots,0)$. More specifically, the above generating function can be rewritten as 
\eq{\label{eq:gen2layers2}
    g_{s,L}(x_1, x_2) = \left[ S_{(2s)}(1,x_1,x_1 x_2) \right]^L\, ,
}
where 
\eq{
    S_{(2s)}(1,x_1,x_1 x_2) = \frac{(1 -x_2)-(1-x_1 x_2)x_1^{2s+1}+ (1-x_1) x_1^{2s+1}x_2^{2s+2}}{(1-x_1)(1-x_2)(1-x_1 x_2)}\,.
}

This generating function for solving the 3d restricted-occupancy problem can be generalized to arbitrary $r$, and its definition is given below:
\begin{definition}
The generating function for the 3d restricted-occupancy problem is
\EQ{
     g_{s,L}(\vec x):= \left[S_{(2s)}\lt(1,x_1, \ldots, x_1 \cdots x_r \rt) \right]^L \, \,,\label{def-genf}
}
where $S_\lambda$ denotes the Schur polynomial associated with a partition $\lambda=(\lambda_1,\lambda_2,\ldots, \lambda_{r+1})$ of $2sL$, satisfying  g $\lambda_i\in\mathbb{Z}_{\geq 0}$, $\lambda_1\geq \lambda_2\geq \dots\geq \lambda_{r+1}$ and $\sum_{i=1}^{r+1}\lambda_i=2sL$,
\EQ{
\label{eq:schur}
    S_\lambda(\{x_i\}_{i=1}^{r+1}):= \frac{\det(x_j^{\lambda_i+n-i})}{\det(x_j^{n-i})}\,.
}
Given a sequence of $r$ integers $\vec{M}:= (M_1,M_2,\ldots,M_r)$ such that $2sL \geq M_1 \geq  \ldots \geq M_r \geq 0$, and let $c_{s,L}(\vec M)$ be the coefficient of $x_1^{M_1} x_2^{M_2}\cdots x_r^{M_r}$ in the expansion of the generating function \eqref{def-genf}, then $c_{s,L}(\vec M)$ counts the ways of assigning $M_a$ boxes to $L$ cells $\{n^{(a)}_{1}, \ldots, n^{(a)}_{L}\}$, subject to the conditions $\sum_{\alpha=1}^L n_\alpha^{(a)} = M_a$, $n_\alpha^{(a)} \leq n_\alpha^{(a-1)}$, and $n_\alpha^{(1)}\leq 2s$. 
\end{definition}
Note that the Schur polynomial $S_\lambda$ with $r$ variables is nothing but the character of the Lie algebra $A_r$ associated with the representation labeled by $\lambda$. In particular, if we specialize to the case when $r=1$,  \eqref{def-genf} will reproduce the generating function for the 2d restricted-occupancy problem,  {\it i.e.,} \eqref{eq:2dgen}. Specializing (\ref{def-genf}) to the $r=2$ case leads to \eqref{eq:gen2layers2}, or equivalently \eqref{eq:gen2layers1}. The restricted-occupancy coefficient $c_{s,L}(\vec{M})$ stands for the number of states with the same weight inside a tensor product.

For the general Schur polynomial defined as \eqref{eq:schur}, the sum of its coefficients can be obtained by taking $x_i$ to 1, namely, we have
\eq{
  S_\lambda(1,\ldots,1) = \prod_{1\leq i < j \leq r+1} \frac{\lambda_i - \lambda_j + j - i}{j-i} \,.
}
Using this identity, one can check that the numbers of spin-chain states add up to \eq{\label{eq:idcomplete}
   \sum_{\vec{M}} c_{s,L}(\vec{M}) = g_{s,L}(\vec x)\Big|_{\{x_i=1\}} = \binom{2s+r}{r}^L \,.
}
This verifies the completeness of the Hilbert space with generic twist angles \eqref{eq:dim-Hilb}.

\subsection{Descendants and a counting formula}
\begin{definition}
Given a root system associated with a Lie algebra $\mathfrak{g}$, let $\alpha_i$ be the simple roots, $\alpha = \sum_{i=1}^r n_i\alpha_i$ be a positive root and $t_i$, $i=1,\ldots, r$, be a set of formal variables. Using the multi-index notation $t^\alpha = t_1^{n_1}\cdots t_r^{n_r}$, the Verma-module character is defined as \cite[(8.87) with $\Lambda=0$]{Ramond_2010}
\eq{
\chi(\mathfrak{g}) :=  \prod_{\alpha \in \Delta^+} (1- t^\alpha)^{-1}\label{eq:character}
\,.}
We will refer to it simply as the character in this paper, and note that its reciprocal $\chi(\mathfrak{g})^{-1}$ is nothing but the Weyl denominator. 
\end{definition}

\begin{definition}
The shift operator $\mathcal{D}_{P(t_1,\ldots,t_r)}$ associated with a polynomial $P(t_1,\ldots,t_r)$ is defined as
\eq{
\mathcal{D}_{P(t_1,\ldots,t_r)} f(\vec M) :=  \sum_{t^\beta \in P(t_1,\ldots,t_r)} \sgn(\beta) f(M_1 - \beta_1, \cdots, M_r - \beta_r) 
\,,
\label{shift}
}
where $\sgn(\beta)$ is the sign in front of $t_1^{\beta_1}\cdots t_{r}^{\beta_r}$ in the $t$-expansion of $P(t_1,\ldots,t_r)$. 
\end{definition}

\paragraph{The counting formula} 

The number of physical Bethe states of the untwisted model is related to that of the twisted model as
\eq{
\mu_{s,L}(\vec M)  =  \mathcal{D}_{\chi(\mathfrak{g})^{-1}} c_{s,L}(\vec M) \,,
\label{thm}
}
Mathematically, this formula expresses the multiplicity $\mu_{\lambda}$ of the occurrence of the irreducible representation $V_\lambda$ in the $L$-fold tensor power of the $2s$-symmetric representation of the Lie algebra $\mathfrak{g} = \mathfrak{su}(r+1)$. We use interchangeably $\mu_\lambda$ and $\mu_{s,L}(\vec M)$ to denote the multiplicity associated with the highest weight. We identify the highest weight \eqref{weight} with the Young diagram  
\eq{
\label{eq:youngdiagram}
    \lambda=(2sL-M_1,M_1-M_2,\ldots,M_{r-1}-M_r,M_r)\,.
}
The Dynkin labels are given by the difference of the length of adjacent rows, $\lambda_i - \lambda_{i+1}$. 

The counting formula tells us how the solutions counted by $c_{s,L}(\vec{M})$ in the case where the $SU(r+1)$ symmetry is broken to $U(1)^r$ recombine into the highest-weight states of $\mathfrak{su}(r+1)$ counted by $\mu_{\lambda}$ in the untwisted limit. In the language of representation theory, the $c_{s,L}(\vec{M})$'s can be understood as the multiplicities of (irreducible) representations of $\mathfrak{u}(1)^r$ appearing in the decomposition of the tensor product of representations of $\mathfrak{su}(r+1)$.

\begin{example}\label{ex:su2}
    For rank 1, $\chi(\mathfrak{su}(2))^{-1}=1-t$. This suggests the formula 
    \eq{
       \mu_{s,L}(M)=c_{s,L}(M)-c_{s,L}(M-1)\,.\label{dc-formula-2}
    }
\end{example}

\paragraph{A physical reasoning:} We first consider the simplest cases of $r=1$ and $r=2$ to explain the idea behind the counting formula. 

For $r=1$, let $\ket{M,i}$ denote the $i$-th eigenstate of the XXX spin chain with $M$ magnons. In the untwisted limit, the spin chain is $SU(2)$ invariant and each $\ket{M,i}$ is a highest-weight state of $SU(2)$,  {\it i.e.,} $S^+\ket{M,i}=0$. When we turn on a generic twist angle $\theta$, the symmetry of the spin chain is broken to $U(1)$, and the $SU(2)$ multiplets of $\ket{M,i}$ split into states carrying the $U(1)$ charge. The Hilbert subspace of a twisted spin chain with fixed $U(1)$ charge $M$ (note that the raising/lowering operator $S^\pm$ respectively carries the charge $\mp 1$) in the untwisted limit is given by 
\eq{
    \{\ket{M,i}\}\cup\{S^-\ket{M-1,i}\}\cup\cdots\cup\{(S^-)^M\ket{0,i}\}\,.\label{untwist-str1}
}
The number of the above states is denoted as $c_{s,L}(M)$ for a spin-$s$ chain with a generic twist angle. Similarly, the Hilbert subspace of charge $M-1$ in the untwisted limit is given by  
\eq{
    \{\ket{M-1,i}\}\cup\{S^-\ket{M-2,i}\}\cup\cdots\cup\{(S^-)^{M-1}\ket{0,i}\}\,.\label{untwist-str2}
}
As the action of $S^-$ is single-valued, the map between $\ket{M,i}$ and $S^-\ket{M,i}$ is one-to-one, then the number of highest-weight states $\ket{M,i}$ is simply given by \eqref{dc-formula-2}.

\paragraph{Example in Bethe ansatz} Let us depict how the Hilbert-space structure in \eqref{untwist-str1} and \eqref{untwist-str2} appears in the untwisted limit and show how the idea works to count the number of highest-weight states as \eqref{dc-formula-2}  with an explicit example of $\mathfrak{su}(2)$ Heisenberg spin chain with $L=4$ and $s=1/2$. We emphasize again, however, the counting formula itself comes purely from the representation theory and does not rely on any assumptions in the Bethe ansatz approach. When we impose a twisted boundary condition with the twisting angle $\theta$ very closed to $0$, e.g. $e^{i\theta}=1.01$, six solutions at $M=2$ are given by (up to normalizations and only three significant digits are shown), 
\begin{align}
\begin{split}
    &\ket{2,1}=\cB(0.289+8.29\times 10^{-4}i)\cB(-0.289+8.29\times 10^{-4}i)\ket{0},\\
    &\ket{2,2}=\cB(0.502i)\cB(-0.498i)\ket{0},\\
    &\widetilde{\ket{2,1}}=\cB(0.500+2.01\times 10^{2}i)\cB(-0.500+3.08\times 10^{-8}i)\ket{0},\\
    &\widetilde{\ket{2,2}}=\cB(-0.500+2.01\times 10^{2}i)\cB(0.500+3.08\times 10^{-8}i)\ket{0},\\
    &\widetilde{\ket{2,3}}=\cB(2.01\times 10^{2}i)\cB(-7.70\times 10^{-9}i)\ket{0},\\
    &\overline{\ket{2,1}}=\cB(1.74\times 10^{2}+3.01\times 10^2i)\cB(-1.74\times 10^{2}+3.01\times 10^2i)\ket{0}.
\end{split}
\end{align}
Recall that when the spectral parameter is very large $u\to\infty$, the lowering operator $\cB(u)$ reduces to $S^-$, ${\cal B}( u)\sim i u^{L-1}S^-$. One can see from the numerical computation that in the untwisted limit $\theta\to 0$, $\ket{2,1}$ and $\ket{2,2}$ become highest-weight states of $\mathfrak{su}(2)$, and 
\begin{align}
    \begin{split}
        &\widetilde{\ket{2,1}}\to S^-\ket{1,1}_{\theta=0},\quad \ket{1,1}_{\theta=0}=\cB(-0.5)\ket{0},\\
        &\widetilde{\ket{2,2}}\to S^-\ket{1,2}_{\theta=0},\quad \ket{1,2}_{\theta=0}=\cB(0.5)\ket{0},\\
        &\widetilde{\ket{2,3}}\to S^-\ket{1,3}_{\theta=0},\quad \ket{1,3}_{\theta=0}=\cB(0)\ket{0},\\
        & \overline{\ket{2,1}}\to (S^-)^2\ket{0},
    \end{split}
\end{align}
where we normalized the states properly to avoid the divergence in the untwisted limit. Similarly, at $M=1$ in the twisted case with $e^{i\theta}=1.01$, we have four solutions, 
\begin{align}
    \begin{split}
        &\ket{1,1}=\cB(-0.500+0.124\times 10^{-3}i)\ket{0},\\
        &\ket{1,2}=\cB(0.500+0.124\times 10^{-3}i)\ket{0},\\
        &\ket{1,3}=\cB(6.22\times 10^{-4}i)\ket{0},\\
        &\widetilde{\ket{1,1}}=\cB(4.02\times 10^2i)\ket{0}.
    \end{split}
\end{align}
$\ket{1,i}$'s ($i=1,2,3$) become the highest-weight state $\ket{1,i}_{\theta=0}$ in the untwisted limit $\theta\to0$, and $\widetilde{\ket{1,1}}\to S^-\ket{0}$. We then see that except for the highest-weight states $\ket{2,i}$'s ($i=1,2$) remaining in the untwist limit, one can construct a one-to-one correspondence between states at $M=2$ and $M=1$, i.e. $\widetilde{\ket{2,i}}\leftrightarrow \ket{1,i}$ ($i=1,2,3$) and $\overline{\ket{2,1}}\leftrightarrow \widetilde{\ket{1,1}}$. The number of highest-weight states of $\mathfrak{su}(2)$ at $M=2$ can then be found by subtracting the number of twisted solutions at $M=2$ by that of $M=1$, 
\begin{equation}
    \mu_{1/2,4}(2)=c_{1/2,4}(2)-c_{1/2,4}(1)=2.
\end{equation}
\begin{flushright}
    $\Box$
\end{flushright}

For $r=2$, the Hilbert subspace spanned by the states with magnon numbers ($M_1,M_2$) in the untwisted limit is explicitly given by 
\EQ{
    &\{\ket{M_1,M_2,i}\}\cup\{E_{-\alpha_1}\ket{M_1-1,M_2,i}\}\cup\{E_{-\alpha_2}\ket{M_1,M_2-1,i}\}\\
    &\cup\{E_{-\alpha_1}E_{-\alpha_2}\ket{M_1-1,M_2-1,i}\}\cup\{E_{-\alpha_1-\alpha_2}\ket{M_1-1,M_2-1,i}\}\\
    &\cup\{E_{-\alpha_1}^2\ket{M_1-2,M_2,i}\}\cup\{E_{-\alpha_2}^2\ket{M_1,M_2-2,i}\}\cup\cdots\,.
}
The total number of states in the above Hilbert subspace is $c_{s,L}(M_1,M_2)$, and one may re-organize it into a tower of states as 
\EQ{
\tikzset{font=\small}
\begin{tikzpicture}[scale=0.8]
    \node at (0,0) {{\color{violet}$\{\ket{M_1,M_2,i}\}$}};
    \node at (-3.25,-1) {{\color{blue}$\{E_{-\alpha_1}\ket{M_1-1,M_2,i}\}$}};
    \node at (3.25,-1) {{\color{red}$\{E_{-\alpha_2}\ket{M_1,M_2-1,i}\}$}};
    \node at (0,-2) {{\color{blue}$\{E_{-\alpha_2}E_{-\alpha_1}\ket{M_1-1,M_2-1,i}\}$}};
    \node at (0,-2.5) {{\color{red}$\{E_{-\alpha_1}E_{-\alpha_2}\ket{M_1-1,M_2-1,i}\}$}};
    \node at (-6.5,-2.25) {{\color{blue}$\{E_{-\alpha_1}^2\ket{M_1-2,M_2,i}\}$}};
    \node at (6.5,-2.25) {{\color{red}$\{E_{-\alpha_2}^2\ket{M_1,M_2-2,i}\}$}};
    \node at (-6.5,-3.25) {\begin{turn}{90}$\ddots$\end{turn}};
    \node at (0,-3) {$\vdots$};
    \node at (6.5,-3) {$\ddots$};
\end{tikzpicture}\label{f:tower}
}
This tower resembles the structure of the highest-weight module of $\mathfrak{su}(3)$. We aim to work out the number $\mu_{s,L}(M_1,M_2)$ of the highest-weight states ${\color{violet}\{\ket{M_1,M_2,i}\}}$ colored in violet from the given information of $c_{s,L}(M_1,M_2)$. The tower is not infinite and has boundaries where the highest-weight state $\ket{m_1,m_2,i}$ corresponding to the component $\{E_{-j_1}E_{-j_2}\cdots E_{-j_n}\ket{m_1,m_2,i}\}$ with $j_i\in\{\alpha_1,\alpha_2\}$ violates the highest-weight condition by either $m_1-2m_2<0$ or $m_2<0$. We assign $\mu_{\lambda}$ with the magnon charges $(m_1,m_2)$ that violates the highest-weight condition to zero,  {\it i.e.,} 
\eq{
    \mu_{s,L}(m_1,m_2):= 0\,,\quad {\rm for}\ m_1-2m_2<0\ {\rm or}\ m_2<0\,.
}
Using the one-to-one correspondence between the states in $\{\ket{m_1,m_2,i}\}$ and $\{E_{-j_1}E_{-j_2}\cdots E_{-j_n}\ket{m_1,m_2,i}\}$, which is true as long as the lowering operators do not annihilate the highest-weight states, we see that
\EQ{
    c_{s,L}(M_1,M_2)&={\cal D}_{\chi(\mathfrak{su}(3))}\mu_{s,L}(M_1,M_2)\\
    &=\sum_{t_1^{\beta_1}t_2^{\beta_2}\in \chi(\mathfrak{su}(3))}\mu_{s,L}(M_1-\beta_1,M_2-\beta_2)\,,\label{eq:d-to-c}
}
where we extended the definition of the shift operator ${\cal D}$ to the formal power series $\chi(\mathfrak{g})$. 
We observe that the descendants respectively colored in blue and red can be sorted into descendant towers of {\color{blue}$\{E_{-\alpha_1}\ket{M_1-1,M_2,i}\}$} and {\color{red}$\{E_{-\alpha_2}\ket{M_1,M_2-1,i}\}$}. In this way, we can decompose the tower into such descendant towers as 
\EQ{
\begin{tikzpicture}
    \node at (0,0) [violet] {1};
    \node at (-1,-1) [blue] {1};
    \node at (1,-1) [red] {1};
    \node at (0,-2) {${\color{blue}1}+{\color{red}1}$};
    \node at (-2,-2) [blue] {1};
    \node at (2,-2) [red] {1};
    \node at (-3,-3) [blue] {1};
    \node at (-1,-3) {${\color{blue}2}+{\color{red}1}-{\color{orange}1}$};
    \node at (1,-3) {${\color{blue}1}+{\color{red}2}-{\color{cyan}1}$};
    \node at (3,-3) [red] {1};
    \node at (-4,-4) [blue] {1};
    \node at (-2,-4) {${\color{blue}2}+{\color{red}1}-{\color{orange}1}$};
    \node at (0,-3.75) {${\color{blue}2}+{\color{red}2}-{\color{orange}1}$};
    \node at (0,-4.25) {$-{\color{cyan}1}+{\color{gray}1}$};
    \node at (2,-4) {${\color{blue}1}+{\color{red}2}-{\color{cyan}1}$};
    \node at (4,-4) [red] {1};
    \node at (-5,-5) [blue] {1};
    \node at (-3,-5) {${\color{blue}2}+{\color{red}1}-{\color{orange}1}$};
    \node at (-1,-4.75) {${\color{blue}3}+{\color{red}2}-{\color{orange}2}$};
    \node at (-1,-5.25) {$-{\color{cyan}1}+{\color{gray}1}$};
    \node at (1,-4.75) {${\color{blue}2}+{\color{red}3}-{\color{orange}1}$};
    \node at (1,-5.25) {$-{\color{cyan}2}+{\color{gray}1}$};
    \node at (3,-5) {${\color{blue}1}+{\color{red}2}-{\color{cyan}1}$};
    \node at (5,-5) [red] {1};
     \node at (-6,-6) [blue] {1};
    \node at (-4,-6) {${\color{blue}2}+{\color{red}1}-{\color{orange}1}$};
    \node at (-2,-5.75) {${\color{blue}3}+{\color{red}2}-{\color{orange}2}$};
     \node at (-2,-6.25) {$-{\color{cyan}1}+{\color{gray}1}$};
    \node at (0,-5.75) {${\color{blue}3}+{\color{red}3}-{\color{orange}2}$};
    \node at (0,-6.25) {$-{\color{cyan}2}+{\color{gray}2}$};
    \node at (2,-5.75) {${\color{blue}2}+{\color{red}3}-{\color{orange}1}$};
    \node at (2,-6.25) {$-{\color{cyan}2}+{\color{gray}1}$};
    \node at (4,-6) {${\color{blue}1}+{\color{red}2}-{\color{cyan}1}$};
    \node at (6,-6) [red] {1};
    \draw[blue] (-0.75,-0.25)--(-6.5,-6);
    \draw[blue] (-0.75,-0.25)--(5,-6);
    \draw[red] (0.75,-0.25)--(6.5,-6);
    \draw[red] (0.75,-0.25)--(-5,-6);
    \draw[cyan,dashed,thick] (1.25,-1.75)--(5.5,-6);
    \draw[cyan,dashed,thick] (-3,-6)--(1.25,-1.75);
    \draw[orange,dashed,thick] (-1.25,-1.75)--(-5.5,-6);
    \draw[orange,dashed,thick] (3,-6)--(-1.25,-1.75);
    \draw[gray,dotted,ultra thick] (0,-2.75)--(3.25,-6);
    \draw[gray,dotted,ultra thick] (0,-2.75)--(-3.25,-6);
\end{tikzpicture}
}
We simplified the diagram into a set of numbers counting the components at each magnon charge $(m_1,m_2)$. The towers colored in {\color{blue}blue}, {\color{red}red}, {\color{orange}orange}, {\color{cyan}cyan} and {\color{gray}gray} are respectively associated with some descendants of the highest-weight states, {\color{blue}$\{\ket{M_1-1,M_2,i}\}$}, {\color{red}$\ket{M_1,M_2-1,i}\}$}, ${\color{orange}\{\ket{M_1-2,M_2-1,i}\}}$, {\color{cyan}$\{\ket{M_1-1,M_2-2,i}\}$} and {\color{gray}$\{\ket{M_1-2,M_2-2,i}\}$}. The number of highest-weight states in violet is thus given by 
\EQ{
{\color{violet}\mu_{s,L}(M_1,M_2)}
&=c_{s,L}(M_1,M_2)-{\color{blue}c_{s,L}(M_1-1,M_2)}-{\color{red}c_{s,L}(M_1,M_2-1)}\\
    &+{\color{orange}c_{s,L}(M_1-2,M_2-1)}+{\color{cyan}c_{s,L}(M_1-1,M_2-2)}-{\color{gray}c_{s,L}(M_1-2,M_2-2)}\\
    &={\cal D}_{\chi(\mathfrak{su}(3))^{-1}}c_{s,L}(M_1,M_2)\,.
\label{su3}
}

One can generalize the above argument to arbitrary rank. A trivial identity reads 
\eq{
    \mu_{s,L}(\vec{M})={\cal D}_{\chi(\mathfrak{g})^{-1}}{\cal D}_{\chi(\mathfrak{g})}\mu_{s,L}(\vec{M})\,.
}
Due to the correspondence between the tower of states \eqref{f:tower} and the representation of $\mathfrak{g}$, one further expect the following relation,  
\eq{
    c_{s,L}(\vec{M})={\cal D}_{\chi(\mathfrak{g})}\mu_{s,L}(\vec{M})\,.
}
It follows that 
\eq{
    \mu_{s,L}(\vec{M})={\cal D}_{\chi(\mathfrak{g})^{-1}}c_{s,L}(\vec{M})\,.
}
A rigorous but technical proof will be presented separately in a mathematical paper \cite{SZZZ}. 

\subsection{Spin 1/2: hook-length formula}

For $\mathfrak{su}(r+1)$ representations, there exists a simple hook-length formula for the multiplicity $\mu_{\lambda}$
that counts the standard Young tableaux (see e.g., \cite[Hook Length Formula 4.12]{fulton1991representation})
\eq{
    \mu_{\lambda}=\frac{|\lambda|!}{\prod_{i=1}^{r+1}h_i!}\prod_{i<j}(h_i-h_j)\,,\label{eq:hook}
}
where $|\lambda|:=\sum_{i=1}^{r+1}\lambda_i$, the hook length $h_i$ associated with the $i$-th row is defined as
\eq{
\label{eq:hookl}
    h_i:=\lambda_i + r+1-i\,.
}
As a consistency check, we compare the above hook-length formula with our formula \eqref{thm} specified to $s=1/2$. 
From Figure~\ref{fig: 3d} it is straightforward to read off $c_{1/2,L}(\vec{M})$ as
\EQ{
    c_{1/2,L}(\vec{M})
    &={L \choose M_1}\prod_{i=1}^{r-1} {M_i \choose M_{i+1}}
    \\
    &=\frac{|\lambda|!}{\prod_{i=1}^{r+1}  \lambda_i!}
    \,.
}
Note that the partition $\lambda$ is given in terms of $L$ and $M_i$'s as in \eqref{eq:youngdiagram} with $s=1/2$, and $|\lambda|=L$. 

Let us first compare with the rank-1 and rank-2 cases via direct computations. For $\mathfrak{g}=\mathfrak{su}(2)$, the multiplicity formula \eqref{dc} may be written as
\EQ{
\mu_{\lambda}
&=\frac{|\lambda|!}{\lambda_1!\lambda_2!}-\frac{|\lambda|!}{(\lambda_1+1)!(\lambda_2-1)!}
 \\
&=\frac{|\lambda|!}{h_1!h_2!}\lt(h_1-h_2\rt)\,.
}
For $\mathfrak{g}=\mathfrak{su}(3)$, we find from \eqref{su3}
\EQ{
\mu_{\lambda}&=\frac{|\lambda|!}{(\lambda_1+2)!(\lambda_2+1)!\lambda_3!}\Big[(\lambda_1+2)(\lambda_1+1)(\lambda_2+1)-(\lambda_1+2)(\lambda_2+1)\lambda_2\\
    &-(\lambda_1+2)(\lambda_1+1)\lambda_3+(\lambda_2+1)\lambda_2\lambda_3+(\lambda_1+2)\lambda_3(\lambda_3-1)-(\lambda_2+1)\lambda_3(\lambda_3-1)\Big]\\
&=\frac{|\lambda|!}{h_1!h_2!h_3!}(h_1-h_2)(h_1-h_3)(h_2-h_3)\,.
}
For $\mathfrak{g}=\mathfrak{su}(r+1)$, we can prove the following proposition:
\begin{proposition}
For $s=1/2$, the hook-length formula \eqref{eq:hook} can be reproduced from the counting formula \eqref{thm}.
\end{proposition}
\begin{proof}
Recall that each term in the counting formula is obtained from a shift associated with a monomial in the Weyl denominator $\chi(\mathfrak{g})^{-1}$, 
\eq{
    \chi\lt(\mathfrak{su}(r+1)\rt)^{-1}=\\(1-t_1)\cdots(1-t_r)(1-t_1t_2)\cdots(1-t_{r-1}t_r)\cdots(1-t_1t_2\cdots t_r)\,,
}
and each $t_i$ in the monomial moves one box from the $(i+1)$-th row to the $i$-th row. We thus see that the shift operator ${\cal D}_{\chi^{-1}}$ generates a permutation action of the symmetric group $\mathfrak{S}_{r+1}$ on the Young diagram.
We have
\EQ{
\label{eq:hookpf1}
    \mathcal{D}_{\chi(\mathfrak{su}(r+1))^{-1}} c_{1/2,L}(\vec M) &= \sum_{\sigma \in \mathfrak{S}_{r+1}} \sgn(\sigma)\frac{\lambda!}{\prod_{i=1}^{r+1} (\lambda_i + \sigma(i) - i)} \\
&=\sum_{\sigma \in \mathfrak{S}_{r+1}}\sgn(\sigma)\frac{\lambda!}{\prod_{i=1}^{r+1} h_i!} \prod_{i=1}^{r+1}\prod_{j=0}^{r-\sigma(i)} (h_i + j)\,,
}
where we used \eqref{eq:hookl} to reach the second line. It remains to show that
\eq{
\sum_{\sigma \in \mathfrak{S}_{r+1}}\sgn(\sigma)\prod_{i=1}^{r+1} \prod_{j=0}^{r-\sigma(i)} (h_i + j)= \prod_{i < j} (h_i - h_j)\,.}
For each $h_i$, the left-hand side is a polynomial of degree $r$. For each pair $i$ and $j$, we may divide the set of permutations into $\sigma$ and $(ij)\sigma$. If $h_i = h_j$, then the two sets cancel with each other and the polynomial vanishes.  We conclude that it must be the Vandermonde polynomial.
\end{proof}

\subsection{Completeness of the Hilbert space}
Having obtained the multiplicity for the irreducible representations, it remains to show that all the states add up to the dimension of the Hilbert space \eqref{eq:dim-Hilb} in the untwisted limit,  {\it i.e.,} 
\eq{
    \dim H=\sum_{\vec{M}}\mu_\lambda(\vec{M})\dim R_{\lambda}\,,
    \label{eq:decompose}
}
where the summation runs over all $\vec{M}$ that corresponds to a Young diagram. We first remark that the completeness of the Hilbert space is intuitively obvious, as the states merely recombine into highest-weight modules. To prove it combinatorially, however, is nontrivial.

For rank 1 and $s=1/2$, the completeness of the Hilbert space may be verified directly
\EQ{
\sum_{M=0}^{L/2} \left[{L \choose M} - {L \choose M-1}\right] (L-2M+1)
&= {L \choose L/2} + 2\sum_{M=0}^{L/2-1} {L \choose M}\\
&= 2^L \,,
\label{su2complete}
}
where we evaluated the telescoping sum in the first step. For $s>1$, the completeness was established in \cite{Kirillov1985} using combinatorial identities. For higher ranks, this is even more nontrivial \cite{Kirillov_1987}. Instead of attempting a full proof, here we present one example to show how the counting formula reproduces the dimension of the full Hilbert space.

\begin{example} For $\mathfrak{g}=\mathfrak{su}(3)$, the decomposition of the Hilbert space of two spin-1 particles is given by 
\eq{
\tiny
\ytableausetup{centertableaux}
    \ydiagram{2}\otimes \ydiagram{2}={\ydiagram{4} \atop {M_1 = 0, M_2 = 0}} \oplus {\ydiagram{3,1} \atop M_1 = 1, M_2 = 0} \oplus {\ydiagram{2,2} \atop M_1 = 2, M_2 = 0}\,,
}
or in terms of the dimensions of the representations, 
${\bf 6}\otimes {\bf 6}={\bf 15}'\oplus {\bf 15}\oplus \overline{\bf 6}
$.
We have labeled the corresponding magnon numbers below each Young diagram. 
Our formula \eqref{su3} shows that $\mu_{1,2}(0,0) = \mu_{1,2}(1,0) =\mu_{1,2}(2,0)= 1$. Note that although the combination of magnon numbers $M_1=2$, $M_2=1$ is allowed as a Young diagram ${\tiny \ydiagram{2,1,1}}$, one can easily check that 
$\mu_{1,2}(2,1)=0$,
and the corresponding BAEs have no physical solution. We note that the representation theory imposes implicit selection rules on the BAEs in this way. 

\end{example}

\section{Partial twists}\label{s:example}

\label{sec:partial}

As described in section \ref{sec:symmetries} for the $SU(r+1)$ case, a partial twist breaks the non-Abelian symmetry to a subgroup. Taking some special combination of the twist angles $\theta_i$ to zero corresponds to restoring the corresponding non-Abelian subgroup symmetry of the full $SU(r+1)$. For a partial twist, one can still write down a counting formula similar to \eqref{thm}. The corresponding character can be written down as follows. Note that the character \eqref{eq:character} is given by a product over all the positive roots of the Lie algebra, and each positive root corresponds to an $\mathfrak{su}(2)$ subalgebra. When $\sum_{i=1}^r\theta_i\ne 0$,  {\it i.e.,} the non-Abelian symmetry associated with the $\mathfrak{su}(2)$ subalgebra of the positive root $\alpha = \sum_{i=1}^r \alpha_i$ is broken, we replace $\prod_{i=1}^r t_i \to 0$ in the character $\chi(\mathfrak{g})$, or equivalently, we remove $\alpha$ from the product of $\chi(\mathfrak{g})$ to obtain the new character for the partial twist.

\subsection{Algebra decomposition and branching rule}
Mathematically, the above partial twist of a spin chain can be understood as a restricted representation. Let us consider the following decomposition
\eq{
\label{eq:decomp}
    \mathfrak{g} \longrightarrow \bigoplus_{j=1}^k \mathfrak{g}_j\, , 
}
where each $\mathfrak{g}_j$ is a subalgebra of $\mathfrak{g}$. $\mathfrak{g}_j$ and $\mathfrak{g}_\ell$ are not necessarily distinct for $j\neq \ell$ so there is no need to write down the multiplicity in the above algebra decomposition. We have $\sum_{j=1}^k \rank{\mathfrak{g}_j} = r$. The positive roots of $\mathfrak{g}$ decompose into a disjoint union of positive roots of its subalgebras as
\eq{
    \Delta^+ \longrightarrow \coprod_{j=1}^k \Delta_j^+ \,.
}
In our convention, when some subalgebra $\mathfrak{g}_{\ell}$ is $\mathfrak{u}(1)$, $\Delta_\ell^+(\mathfrak{u}(1))$ is an empty set. One can understand the above decomposition by specifying a choice of positive roots that are preserved by the subalgebras. After the decomposition, we still use the same labeling of roots in $\Delta^+$ for the roots preserved in $\coprod_{j=1}^k \Delta_j^+$. This will help us keep track of different decompositions: even though the decomposition of the algebra could be the same, the decomposition of roots and so the restricted representation of the subalgebras would be different. This point will be clearer in the example of $\mathfrak{su}(3)$ in Section~\ref{s:su3example}. 

Given a subset of positive roots $D^+\subset \Delta^+$, there is a smallest decomposition of $\Delta^+$ such that 
\eq{
    D^+ \subset \coprod_{j=1}^k \Delta_j^+ \subset \Delta^+\,.
}
In other words, the decomposition $\coprod_{j=1}^k \Delta_j^+$ determined by $D^+$ can be defined as the intersection of all decompositions of $\Delta^+$ that contain $D^+$. 
\begin{example}[label=ex:decomp]
Consider $\mathfrak{g}=\mathfrak{su}(r+1)$ and choose $D^+ = \{\alpha_2 + \alpha_3, \alpha_4\}$, then the decomposition of $\Delta^+$ determined by $D^+$ is given as 
\eq{
    \Delta^+ \longrightarrow \Delta_1^+ = \{\alpha_2 + \alpha_3, \alpha_4, \alpha_2+\alpha_3+\alpha_4\}\,,
}
where one can see that $\Delta_1^+ \cong \Delta^+(\mathfrak{su}(3))$, indicating  the following decomposition 
\eq{
\mathfrak{su}(r+1) \longrightarrow \mathfrak{su}(3)\oplus \mathfrak{u}(1)^{r-2}\,.
}
\end{example}

The given subset $D^+$ corresponds to a partial twist for the spin chain. More specifically, for each element $\alpha = \sum_i \alpha_i$ in $D^+$, the linear combination of  $\sum_i \theta_i$ vanishes. In particular, if a simple root $\alpha_i$ appears in $D^+$, then $\theta_i = 0$.

Furthermore, from the point of view of the character, the decomposition can also be understood as keeping $t^\alpha$ in the expression if $\alpha\in\coprod_{j=1}^k\Delta_j^+$, while setting $t^{\alpha^\prime} = 0$ if $\alpha^\prime \in \Delta^+-\coprod_{j=1}^k\Delta_j^+$. The original character $\chi(\mathfrak{g})$ decomposes into a product of characters of subalgebras
\eq{
    \chi(\mathfrak{g}) \longrightarrow \prod_{j=1}^k \chi(\mathfrak{g}_j)\, .
}
If some $\Delta_\ell^+$ is empty, then the corresponding algebra will be $\mathfrak{g}_\ell = \mathfrak{u}(1)$, and $\chi(\mathfrak{u}(1)) = 1$. 

As we have discussed for $\mathfrak{g}=\mathfrak{su}(r+1)$, the Bethe states of the periodic spin chain are the highest-weight states. Once some twist parameters are turned on, the algebra $\mathfrak{g}$ decomposes to its subalgebras $\oplus_i \mathfrak{g}_i$. An immediate result is that some original descendants of $\mathfrak{g}$ would become the highest-weight states of $\oplus_i \mathfrak{g}_i$. The physical Bethe states are now counted by the branching coefficients of the restricted irreducible representations. Let us focus on $\mathfrak{g} = \mathfrak{su}(r+1)$ in the rest of this section. To discuss the restricted irreducible representation for the algebra decomposition (\ref{eq:decomp}), we first need to work out the Young diagrams labeling these restricted representations of $\oplus_{j=1}^k \mathfrak{g}_j$. 
\paragraph{Branching rule}
Recall that the original Young diagram $\lambda$ is determined by the spin-chain data $(s, L,\vec{M})$ as 
\eqref{eq:youngdiagram}. For each positive root $\alpha = e_i - e_j$ in $D^+$, we pick the $i$-th and $j$-th rows of the Young diagram and combine them into a two-row Young diagram. 
For any pair of two-row Young diagrams sharing the same row, glue them together in the same order as in $\lambda$.  Repeat this procedure until any two distinct Young diagrams do not have rows in common. 
We end up with a collection of smaller Young diagrams extracted from $\lambda$, which we call $\lambda^{D^+}$. There could be some rows in the Young diagram $\lambda$ that would not be picked in the first step. These rows become the one-row Young diagrams that label the representations of $\mathfrak{u}(1)$'s.
Each element $\lambda^{(i)}\in \lambda^{D^+}\cup \{\lambda_\ell \mid \text{unpicked rows in $\lambda$}\}$ labels an irreducible representation of $\mathfrak{g}_i$. In the end, the restricted representation of $\bigoplus_{j=1}^k \mathfrak{g}_j$ is denoted by a collection of Young diagrams 
\eq{\label{eq:inducedyoung}
    \Lambda=\big(\lambda^{(1)},\ldots, \lambda^{(k)} \big)\,,\quad \text{with $\lambda^{(i)}\in \lambda^{D^+}\cup \{\lambda_\ell \mid \text{unpicked rows in $\lambda$}\}$}\,.
}

\begin{example}[continues=ex:decomp] Consider again the decomposition specified by $D^+= \{\alpha_2+\alpha_3,\alpha_4\}$. 
Following the branching rule, we obtain the final Young diagram $\lambda^{\{\alpha_2+\alpha_3,\alpha_4\}} = (M_1-M_2,M_3-M_4,M_4-M_5)$ labeling the restricted representations. See Figure~\ref{fig:extract}.
\begin{figure}[!h]
\centering
\includegraphics[scale=0.8]{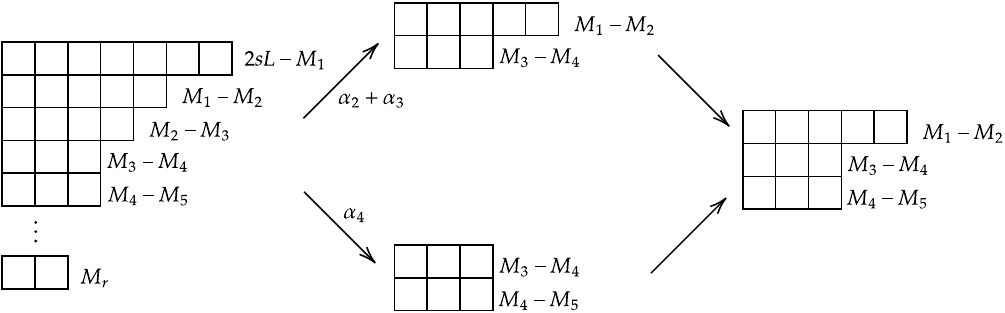}
\caption{An example of the branching rule for the positive roots $D^+ = \{\alpha_2+\alpha_3, \alpha_4\}\subset \Delta^{+}(\mathfrak{su}(r+1))$. In the first step, one extracts a pair of two-row Young diagrams associated with $\alpha_2+\alpha_3$ and  $\alpha_4$ respectively. In the second step, the two Young diagrams are glued along the common row.}
\label{fig:extract}
\end{figure}
\end{example}

Similar to the counting formula \eqref{thm}, the number of physical Bethe states of the partially twisted model is related to that of the twisted model as
\eq{
\mu^{D^+}_{\Lambda}  =  \mathcal{D}_{\prod_{j=1}^k \chi(\mathfrak{g}_j)^{-1}} c_{s,L}(\vec M) \,.
\label{corollary1}
}
Mathematically, it is the multiplicity of the occurrence of the restricted representation $V_{{\Lambda}}$ of the Lie subalgebra $\bigoplus_{j=1}^k \mathfrak{g}_j$. The Young diagram ${\Lambda}$ is a collection of smaller Young diagrams defined as \eqref{eq:inducedyoung}, which are specified by $D^+$ following the branching rule.

    The fully twisted case can be viewed as a special case of the partially twisted case, in the sense that the subset $D^+$ specifying the decomposition is empty. When $D^+=\varnothing$, the algebra $\mathfrak{su}(r+1)$ is broken to its maximal Cartan subalgebra and the Young diagrams ${\Lambda}$ are the collection of one-row diagrams extracted from $\lambda$, labeling the representations of the $\mathfrak{u}(1)$'s. In this case, the shift operator is the identity, and so $\mu^{\varnothing}_{{\Lambda}}  = c_{s,L}(\vec M)$. The completeness can be easily verified by \eqref{eq:idcomplete}. The corresponding generating function for $\mu^{\varnothing}_{{\Lambda}}$ is just \eqref{def-genf}.

The completeness of the Hilbert space is again expected as 
\eq{
    \dim H=\sum_{\vec{M}}\mu^{D^+}_{{\Lambda}}\dim R_{{\Lambda}}\,.
}
The summation now runs over all configurations of the magnon numbers $\vec{M}$ such that each component of ${\Lambda}=(\lambda^{(1)},\ldots, \lambda^{(k)} )$ labels a representation of the unbroken symmetry subalgebra. For each ${\Lambda}$, the representation is the tensor product $R_{{\Lambda}} = R_{\lambda^{(1)}}\otimes R_{\lambda^{(2)}}\otimes \cdots \otimes R_{\lambda^{(k)}}$ and so 
\eq{
    \dim R_{{\Lambda}} = \dim R_{\lambda^{(1)}} \times \dim R_{\lambda^{(2)}} \times \cdots \times \dim R_{\lambda^{(k)}}\,.
}
Whenever a component $\lambda^{(i)}$ is a one-row Young diagram, it labels a representation of $\mathfrak{u}(1)$ and so $\dim R_{\lambda^{(i)}} = 1$. For simplicity of notation, we omit writing out these one-row Young diagrams explicitly.

In the remainder of this section, we provide two examples of partial twists in the rank-2 and rank-3 cases to illustrate this phenomenon in more detail.

\subsection{Partial twists in $SU(3)$}
\label{s:su3example}
We start with the spin chain with an $SU(3)$ global symmetry,  {\it i.e.,} it has the periodic boundary condition with $\theta_1 = \theta_2 = 0$. Its character is
\eq{
\chi(\mathfrak{su}(3)) = \frac{1}{(1-t_1)(1-t_2)(1-t_1 t_2)} \,.
}

Firstly, consider $\theta_1\ne 0$ with $\theta_2=0$, which corresponds to $D^+ = \{\alpha_2\}$. This implies the decomposition $\mathfrak{su}(3) \rightarrow \mathfrak{u}(1)\oplus \mathfrak{su}(2)$. The character of this partial twist is reduced to 
\eq{
\chi(\mathfrak{su}(3)) \longrightarrow \chi(\mathfrak{su}(2)) = \frac{1}{1-t_2}\,.
}
It follows from \eqref{corollary1} that
\eq{
    \mu^{\{\alpha_2\}}_{{\Lambda}}=c_{s,L}(M_1,M_2)-c_{s,L}(M_1,M_2-1)\,,
}
where ${\Lambda}=(\lambda^{(1)},\lambda^{(2)})$ with $\lambda^{(1)} = (2sL-M_1)$ and $\lambda^{(2)} = (M_1-M_2,M_2)$. See the upper-right corner of Figure~\ref{fig:su3youngdiagrams}. The dimension of the $ \mathfrak{u}(1) \oplus \mathfrak{su}(2)$ representation is
\eq{
    \dim R_\Lambda = \dim R_{\lambda^{(1)}}^{\mathfrak{u}(1)}\times  \dim R_{\lambda^{(2)}}^{\mathfrak{su}(2)} = M_1-2M_2 +1\,.
}
The completeness of the Hilbert space implies
\eq{
    \dim H=\sum_{\Lambda}\mu^{\{\alpha_2\}}_\Lambda  \dim R_\Lambda\,,
}
where the sum is over all admissible Young diagrams. 
For example, when $s=1$, $L=2$, it reads  
\EQ{
{\bf 6}\otimes {\bf 6} &= ({\bf 5}, 4) + 2 ({\bf 4}, 1) + 3({\bf 3}, -2) + ({\bf 3},4) + 2({\bf 2}, -5) + 2({\bf 2}, 1) \\
&+ ({\bf 1}, -8) + ({\bf 1}, -2) + ({\bf 1}, 4)\,,
}
where the $\mathfrak{u}(1)$ charge in the above equation is chosen to be $4-3\lambda^{(1)}$ to match with the convention in LieART. 
For $s=1/2$, one may prove completeness using binomial identities. Note that
\eq{
    \mu^{\{\alpha_2\}}_\Lambda = \binom{L}{M_1}\binom{M_1}{M_2} - \binom{L}{M_1}\binom{M_1}{M_2-1}\,.
}
It follows by the identity \eqref{su2complete} that
\eq{\sum_{M_1 = 0}^L \sum_{M_2 = 0}^{M_1/2} \mu^{\{\alpha_2\}}_\Lambda (M_1-M_2+1) = \sum_{M_1 = 0}^L {L \choose M_1} \, 2^{M_1} = 3^L  
\label{mu2}\,.}

Next, consider turning on $\theta_2\ne 0$ and keeping $\theta_1=0$. This corresponds to $D^+ = \{\alpha_1\}$.
The branching coefficient is 
\eq{\label{eq:su3ex1}
    \mu^{\{\alpha_1\}}_{{\Lambda}}=c_{s,L}(M_1,M_2)-c_{s,L}(M_1-1,M_2)\,,
}
where $\lambda^{(1)}=(2sL-M_1,M_1-M_2)$ and $\lambda^{(2)}=(M_2)$. See the upper-left corner of Figure~\ref{fig:su3youngdiagrams}.  The representation for $\Lambda$ has dimension
\eq{
    \dim R_{\Lambda} = 2sL-2M_1 +M_2 + 1\,.
}
For $s=1/2$, one can again prove the completeness of the Hilbert space. The branching coefficient reads 
\eq{
    \mu^{\{\alpha_1\}}_{{\Lambda}} = \binom{L}{M_1} \binom{M_1}{M_2} - \binom{L}{M_1-1}\binom{M_1-1}{M_2}\,.\label{mu1}
}
We may apply the identity
\eq{
{n \choose h}{n-h \choose k} = {n \choose k}{n-k \choose h} = {n \choose h+k} {h+k \choose h} \,,
\label{identity}
}
to rearrange \eqref{mu1} to the same form as \eqref{mu2}
\EQ{
&\sum_{M_2=0}^L \sum_{M_1 = 0}^{(L+M_2)/2} \mu^{\{\alpha_1\}}_\Lambda (L-2M_1 + M_2+1) \\ 
&=\sum_{M_2=0}^L \sum_{M_1 = 0}^{(L+M_2)/2}{L \choose M_2} \left[{L-M_2 \choose M_1-M_2} - {L-M_2 \choose M_1-M_2-1}\right] (L-2M_1 + M_2+1)  \\
&= 3^L \,.
}

\begin{figure}[!h]
\centering
\includegraphics[scale=0.8]{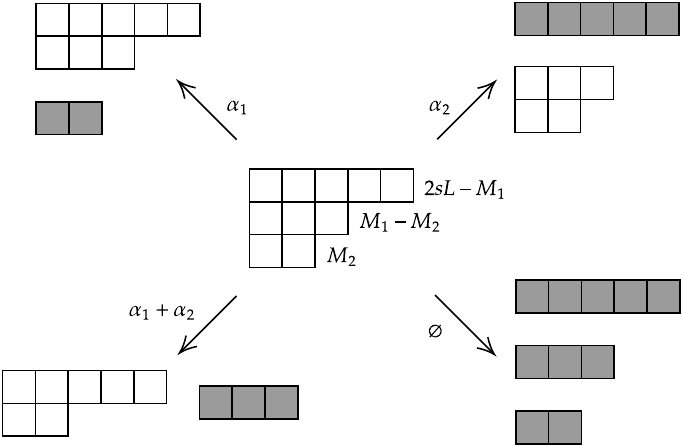}
\caption{An illustration of the partial twists for $\mathfrak{su}(3)$ and the corresponding Young diagrams from the original Young diagram.}
\label{fig:su3youngdiagrams}
\end{figure}

In the case of $\theta_1+\theta_2=0$, only the symmetry associated with the $\mathfrak{su}(2)$ subalgebra of $\alpha_1+\alpha_2$  remains unbroken, namely $D^+ = \{\alpha_1+\alpha_2\}$. We then have 
\eq{
    \mu^{\{\alpha_1+\alpha_2\}}_{{\Lambda}}=c_{s,L}(M_1,M_2)-c_{s,L}(M_1-1,M_2-1)\,,
}
where $\lambda^{(1)} = (2sL-M_1,M_2)$ and $\lambda^{(2)} = (M_1-M_2)$. See the lower-left corner of Figure~\ref{fig:su3youngdiagrams}. The dimension of the representation labeled by $\Lambda$ is
\eq{
    \dim R_{\Lambda}  = 2sL-M_1-M_2+1\,.
}
For $s=1/2$, 
\eq{
    \mu^{\{\alpha_1+\alpha_2\}}_\Lambda = \binom{L}{M_1}\binom{M_1}{M_2} - \binom{L}{M_1-1}\binom{M_1-1}{M_2-1}\,.
}
One can apply \eqref{mu1} with $M_2$ replaced by $M_2^\prime :=  M_1 - M_2$ to verify the completeness condition.

\subsection{Partial twists in $SU(4)$}
The partial twists in $SU(4)$ are even more interesting, as novel patterns of symmetry breaking may appear. Let us consider the example of $\theta_2\neq 0$ and $\theta_1=\theta_3=0$. This is equivalent to choosing $D^+ = \{\alpha_1,\alpha_3\}$ and specifying the following decomposition:
\eq{
    \Delta^+(\mathfrak{su}(4)) \longrightarrow \Delta^+(\mathfrak{su}(2)) \sqcup \Delta^+(\mathfrak{su}(2))\, .
}
The remaining non-Abelian symmetry of the spin chain is expected to be $SU(2)\times SU(2)$, whose character is the product of two $\mathfrak{su}(2)$ characters
\eq{
    \chi(\mathfrak{su}(2)\oplus\mathfrak{su}(2)) = \frac{1}{1-t_1} \frac{1}{1-t_3}\,.
}
The counting formula \eqref{corollary1} gives 
\EQ{
    \mu_{\Lambda}^{\{\alpha_1,\alpha_3\}}=\ & c_{s,L}(M_1,M_2,M_3)-c_{s,L}(M_1-1,M_2,M_3)-c_{s,L}(M_1,M_2,M_3-1)\\
    &+c_{s,L}(M_1-1,M_2,M_3-1)\,,
}
where $\lambda^{(1)}=(2sL-M_1,M_1-M_2)$ and $\lambda^{(2)}=(M_2-M_3,M_3)$. We note that the number of Cartan generators is always preserved after the symmetry breaking. Therefore in the process of $SU(4)$ broken to $SU(2) \times SU(2)$, one can find a linear combination of three Cartan generators in $\mathfrak{su}(4)$ to generate an independent $U(1)$ symmetry. Given two Young diagrams $\lambda^{(1)}$ and $\lambda^{(2)}$ labeling the representations of the two $SU(2)$'s, the $U(1)$ can be fixed as the difference between their sizes, $\lambda^{(1)}-\lambda^{(2)}$. 

\section{Generalizations}
\label{s:gen}

\subsection{Different spins at each site: Kondo-type models}

As the derivation of the formula \eqref{thm} does not require all the sites to have the same representation, one can modify the generating function defined in \eqref{def-genf} to
\EQ{
     \tilde{g}_{\vec{\lambda},L}(\vec x):= \prod_{n=1}^LS_{\lambda^{(n)}}\lt(1,x_1, \ldots, x_1 \cdots x_r\rt) \, \,,\label{def-ggen}
}
to define a modified coefficient $\tilde{c}_{\vec{\lambda},L}(\vec{M})$, and the multiplicity $\tilde{\mu}_\lambda$ is still expected to be computed from the same formula \eqref{thm}.

When all but one sites have the same spin $s$, one can view the site with a different spin $s'$ as an impurity and it gives a prototypical model of the Kondo effect \cite{PhysRevB.25.5935}. The number of solutions of such a twisted Kondo model can be worked out from the spin chain as 
\eq{
    c^{\rm K}_{s,s',L}(M)=\sum_{i=0}^{2s'}c_{s,L}(M-i)\,,
}
where we denoted the number of sites with spin $s$ as $L$, and set the spin of the impurity to be $s'-i$, then summed over all possibilities to obtain $c^{\rm K}_{s,s'}(M)$. One can easily confirm that $\tilde{c}_{\vec{\lambda}_{s,s'},L+1}(M)=c^{\rm K}_{s,s',L}(M)$ for an $(L+1)$-tuple Young diagram $\vec{\lambda}_{s,s'}:= \lt((2s'),(2s),(2s),\ldots,(2s)\rt)$, since $S_{(2s')}(\{1,x\})=\sum_{i=0}^{2s'}x^i$. It can be straightforwardly generalized to higher rank. By using 
\eq{
    S_{(2s')}(1,x_1, \ldots, x_1 \cdots x_r)=\sum_{i_1=0}^{2s'}\sum_{i_2=0}^{2s'-i_1}\cdots\sum_{i_r=0}^{2s'-\sum_{j=1}^{r-1}i_j}\prod_{k=1}^rx_k^{i_k}\,,
}
one obtains the number of solutions in higher-rank Kondo models as 
\eq{
    c^{\rm K}_{s,s',L}(\vec{M})=\sum_{i_1=0}^{2s'}\sum_{i_2=0}^{2s'-i_1}\cdots\sum_{i_r=0}^{2s'-\sum_{j=1}^{r-1}i_j}c_{s,L}(\vec{M}-\vec{i})\,,
}
where $\vec{i}:= (i_1,i_2,\ldots,i_r)$. Similar combinatorial problems have been studied in \cite{Dafnis:2007}.

\subsection{Extension to Lie superalgebras}\label{s:super}

The counting formula \eqref{thm} does not rely on any specific structure of Lie algebras of $A$ type, we expect that it can also be applied to count the physical solutions in spin chains of other types of Lie algebras or even Lie superalgebras. 

In the cases of Lie superalgebras, we replace the character $\chi(\mathfrak{g})$ in \eqref{thm} with that of the Lie superalgebra studied in \cite{kac1977characters}. The positive roots $\Delta^{+}$ consists of two parts
\eq{
    \Delta^+ = \Delta_0^+ \cup \Delta_1^+\,,
}
where $\Delta_0^+$ is the set of positive even roots and $\Delta_1^+$ is the set of positive odd roots. For $\mathfrak{g} = \mathfrak{sl}(m|n)$, the positive roots can be written in the standard basis \cite{frappat:hal-00376660},
\eq{
\begin{aligned}
    & \Delta_0^+ = \left\{e_i - e_j | i<j,\ i,j = 1,\ldots,m \right\} \cup \left\{f_k - f_\ell | k<\ell,\ k,\ell = 1,\ldots,n  \right\}\, , \\
    & \Delta_1^+ = \left\{e_i - f_k | i = 1,\ldots,m,\ k = 1,\ldots, n \right\}\,.
\end{aligned}
}
An observation is that $\Delta_1^+$ can be generated by one element of $\Delta_1^+$, e.g., $e_m-f_1$, by its linear compositions with the elements in $\Delta_0^+$. Therefore, we will choose the simple roots of $\mathfrak{sl}(m|n)$ as follows, which is also called the {\it distinguished simple root system},
\eq{
\alpha_i = e_i - e_{i+1}\,,  \quad \beta_\ell = f_\ell - f_{\ell+1}\,, \quad  \delta = e_{m} - f_1\,,
}
for $i=1,\ldots,m-1$ and $\ell = 1,\ldots, n-1$. The character of the Lie superalgebra $\mathfrak{sl}(m|n)$ is\footnote{There is another character for Lie superalgebras, called the {\it supercharacter}  \cite{kac1977characters}, defined as
$$ {\rm sch}(\mathfrak{g}) = \frac{\prod_{\alpha\in \Delta_1^+} (1-t^\alpha)}{\prod_{\alpha\in \Delta_0^+} (1-t^\alpha)}\,. $$
}
\eq{\label{eq:chsuper}
\chi(\mathfrak{sl}(m|n)) = \frac{\prod_{\alpha\in \Delta_1^+} (1+t^\alpha)}{\prod_{\alpha\in \Delta_0^+} (1-t^\alpha)}\,,
}
where $t^\alpha :=  \big(\prod_i t_i^{m_i }\big) \big(\prod_\ell t_{\ell+m-1}^{n_{\ell}}\big) t_{m+n-1}^{k}$ for some (even or odd) positive root $\alpha = \sum_i m_i \alpha_i + \sum_\ell n_\ell \beta_\ell + k \delta$ and $t_I$ with $I=1,\ldots, m+n-1$ are formal variables associated with each simple root. 

We conjecture that one can follow the definition of the shift operator \eqref{shift} and generalize the counting formula to the case of a Lie superalgebra $\mathfrak{sl}(m|n)$. In general, $\chi(\mathfrak{sl}(m|n))^{-1}$ is not a polynomial. Therefore, one should use its formal Taylor expansion in $t_i$'s to define the corresponding shift operator. We illustrate this by two simplest examples, $\mathfrak{sl}(1|1)$ and $\mathfrak{sl}(1|2)$, in the remainder of this section.

\subsubsection{$\mathfrak{sl}(1|1)$}
The simplest case, $\mathfrak{sl}(1|1)$, describes the XX chain with a magnetic field \cite{Saleur:1989nw}. Counting solutions in the twisted case is also described by the restricted-occupancy problem shown in Figure~\ref{fig: 2d}. Therefore, we will use the same notation $c_{s,L}(M)$ as in Section~\ref{s:counting}.

There is only one positive root which is odd,
\eq{
    \Delta^+(\mathfrak{sl}(1|1)) = \{e - f\}\,,
}
and the character is 
\eq{
\chi\lt(\mathfrak{sl}(1|1)\rt) = 1+t\,.    
}
The shift operator associated with its reciprocal can also be defined by a slight generalization of \eqref{shift}. We note that $1/(1+t)$ can be formally expanded in $t$ as $1/(1+t) = \sum_{i}(-1)^i t^i$, and the shift operator for each monomial in this expansion just follows \eqref{shift}. We then define
\eq{
\mathcal{D}_{\chi\lt(\mathfrak{sl}(1|1)\rt)^{-1}} :=  \sum_{i=0}^\infty (-1)^i \mathcal{D}_{t^i}\,.
}
Therefore we find the multiplicity 
\EQ{
    \mu_{\lambda}
    =\sum _{i=0}^{\infty} (-1)^i c_{s,L}(M-i)=(-1)^M\sum_{i=0}^M(-1)^i c_{s,L}(i)\,.
\label{conj-sl11}
}
For $s=1/2$, one can prove by induction that
\eq{
    \mu_{\lambda} = {L-1 \choose M}=\frac{(L-1)!}{M!(L-1-M)!} \,.\label{d-sl11h}
}
We observe that the multiplicity can also be read from the hook-length formula 
\eqref{eq:hook} of a Young diagram $\lambda=(L-M,1^M)$. See Figure~\ref{fig: t-JYoung}. 
One can verify the completeness of the Hilbert space as
\eq{
\sum_{M=0}^{L-1} \mu_{\lambda} \times 2 = 2^L \,.
}
Note that $M$ is allowed to take values from $0$ to $L-1$. 

\begin{figure}
        \centering
        \includegraphics[scale=0.85]{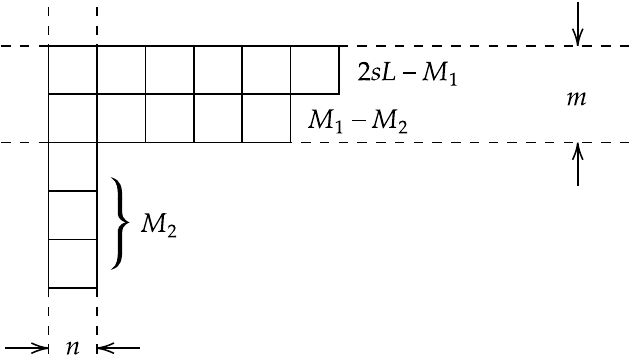}
        \caption{The highest-weight states for $\mathfrak{su}(m|n)$ correspond to Young diagrams that fit inside the union of a horizontal strip of width $m$ and a vertical strip of width $n$ \cite{Bars:1982se}.}
        \label{fig: t-JYoung}
\end{figure}

\subsubsection{$\mathfrak{sl}(1|2)$}
An $\mathfrak{sl}(1|2)$ symmetry underlies the one-dimensional t-J model \cite{Sarkar_1990}, The model describes spin-hopping interaction of $M_1$ electrons in an $L$-site lattice with $M_2$ spin-down excitations \cite{Chao_1977}. The number of states in the twisted spin-$s$ generalization is \cite{Shu:2022vpk}
\eq{
    c^{\rm tJ}_{s,L}(M_1, M_2) = \binom{L}{M_1}c_{s,M_1}(M_2)\,.\label{c-t-J}
}
The positive roots are written in the standard basis as
\eq{\label{eq:sl12roots}
    \Delta_0^+ = \{f_1 - f_2\}\, \,,\quad  \Delta_1^+ = \{e - f_1, e - f_2\}\,.
}
Therefore, the character is given by
\eq{\label{eq:sl12character}
    \chi\lt(\mathfrak{sl}(1|2)\rt) = \frac{\left(1 + t_2\right) \left(1 + t_1 t_2\right)}{1-t_1} \,.
}
In analogy with \eqref{thm}, we have 
\eq{
    \mu_{\lambda} = 
    \sum _{i=0}^{\infty} (-1)^i \lt[c^{\rm tJ}_{s,L}(M_1,M_2-i)-c^{\rm tJ}_{s,L}(M_1-i-1,M_2-i)\rt]\,.
}
For $s=1/2$, we may apply the binomial identity \eqref{identity} and a similar summation trick as in \eqref{d-sl11h} to obtain 
\EQ{
    \mu_{\lambda} &= \sum_{i=0}^\infty (-1)^i\lt[\binom{L}{M_1}\binom{M_1}{M_2-i}-\binom{L}{M_1-M_2-1}\binom{L-M_1+M_2+1}{M_2-i}\rt]\\
    &=\binom{L}{M_1}\binom{M_1-1}{M_2}-\binom{L}{M_1-M_2-1}\binom{L-M_1+M_2}{M_2} \\
&= \frac{L! (L-2 M_1+M_2+1)}{M_1 M_2! (L-M_1)! (M_1-M_2-1)! (L-M_1+M_2+1)} \,.
\label{eq:tJmu}
}
This again agrees with the hook-length formula \eqref{eq:hook} applied to the Young diagram $\lambda=(L-M_1,M_1-M_2,1^{M_2})$. It is the same as $Z(M_1,M_2)$ defined by \cite[Equation~(15)]{Foerster:1992ud}:
\eq{
\label{eq:tJmuidentity1}
\begin{aligned}
    \mu_{\lambda} &= \frac{L-2M_1+M_2+1}{L-M_1+M_2+1} \binom{L}{M_1}\binom{M_1-1}{M_2}\\
	& = \frac{L-2M_1+M_2+1}{L-M_1+M_2+1}\binom{M_1-1}{M_2} \sum_{q=0}^{M_1} \binom{L+M_2 -M_1 +1}{q}\binom{M_1 - M_2 - 1}{q - M_2-1}
    \\
	&=\frac{L-2M_1+M_2+1}{L-M_1+M_2+1}\sum_{q=0}^{M_1} \binom{L+M_2 - M_1 + 1}{q } \binom{M_1 - 1}{q-1}\binom{q-1}{M_2}
    \, , 
\end{aligned}
}
where we applied the Chu-Vandermonde identity
\eq{
    \sum_{j=0}^{k} \binom{m}{j}\binom{n-m}{k-j} = \binom{n}{k}\,,
}
and the identity \eqref{identity}.
The dimension of the representation labeled by such a Young diagram is given by \cite{Foerster:1992ud}
\eq{
    \dim R_\lambda^{\mathfrak{sl}(1|2)}=\lt\{\begin{array}{cc}
    2L+1 & M_1=M_2=0\\
    4(L-2M_1+M_2+1) & {\rm otherwise}
    \end{array}\rt.\,.
}
One can follow \cite[section 6]{FOERSTER1993611} and check the completeness as
\eq{
\sum_{\lambda}\mu_{\lambda} \dim R_\lambda^{\mathfrak{sl}(1|2)} = 3^L\,. 
}

\section{Discussions}

In this paper, motivated by the number of physical Bethe states in integrable models with twisted boundary conditions, we derived a counting formula for the tensor-product multiplicities for the composition of $L$ spins. We focused on the symmetry-breaking pattern of the closed spin chains with an $SU(r+1)$ symmetry by imposing a twisted boundary condition. When all the twist angles are generic and non-vanishing, the symmetry of the spin chain is completely broken to $U(1)^r$. The number of states with a given charge can be found from the generating functions that are products of Schur polynomials. In particular, when we take a $2s$-symmetric representation at each spin site, the counting problem is equivalent to a 3d restricted-occupancy problem. When the symmetry is fully restored with all the twist angles set to zero, the counting formula relates the tensor-product multiplicities $\mu_{\lambda}$ to the restricted-occupancy coefficients $c_{s,L}(\vec{M})$. For $s=1/2$, we proved that the proposed formula reproduces the well-known hook-length formula for tensor-product multiplicities. We further showed that the multiplicity in partially twisted models can also be computed from $c_{s,L}(\vec{M})$ by modifying the Verma-module character accordingly. For $SU(3)$, $s=1/2$, we proved the completeness of the Hilbert space. The counting formula is expected to hold for all Lie algebras, and we also examined its applicability to Lie superalgebras. Interestingly, for $s=1/2$ our formula also reduces to the hook-length formula for the Lie superalgebras  $\mathfrak{sl}(1|1)$ and $\mathfrak{sl}(1|2)$. 

Results for the untwisted and partially twisted cases are obtained purely from representation theory, and it remains a non-trivial task to reproduce the multiplicities from the number of physical solutions to BAEs of higher-rank spin chains, Kondo models, and t-J models. For $s\geq 1$ and also in higher-rank models, there are not only singular physical solutions but also repeated and singular repeated solutions \cite{Hao:2013rza}, whose presence requires imposing further physical conditions on the untwisted or partially twisted systems. It is generally complicated to solve the BAEs directly with additional (frequently not fully understood) physical constraints to identify all such physical solutions.  It is now widely believed that the BAEs contain all the eigenstates and some extra non-physical solutions \cite{Hao:2013jqa,Nepomechie:2013mua,Nepomechie:2014hma}.  As an interesting recent discovery, the rational $Q$ system \cite{Marboe:2016yyn}, designed to solve the BAEs efficiently, turned out to find precisely the physical Bethe roots. The physical conditions and also the $Q$ system become much more complicated for higher-spin and higher-rank spin chains (see e.g., \cite{Hao:2013rza,Hou:2023ndn}). The first step towards a complete algorithm to solve for only physical states in general spin chains is to clarify the number of physical solutions associated with different magnon charges. It is expected that the representation-theoretic structure can be encoded in the $Q$ system and the Hirota-equation approaches \cite{Krichever:1996qd,Kazakov:2015efa} to the BAEs. We hope that the counting formula developed in this paper will guide the development of a complete algorithm to solve the BAEs and the study of supersymmetric gauge theories on singular loci.

In this paper, we focused on the connection between the counting formula and the closed spin chains. Since the idea is purely algebraic, we expect similar arguments to be applicable even to open spin chains, whose BAEs with diagonal and off-diagonal boundaries can be derived following the techniques developed in \cite{Sklyanin:1988yz,Cao:2013nza}. The Bethe/Gauge correspondence for open spin chains has been formulated at the level of equations,  {\it i.e.,} the map between the BAEs and the vacua equations in the gauge theory has been constructed \cite{Kimura:2020bed,Ding:2023auy,Ding:2023lsk,Ding:2023nkv,Wang:2024zcr}. One of the problems in this correspondence is that there are different ways to realize the vacua equations of the same gauge theory. It is therefore important to look into the details of the solutions and realize the Seiberg-like dualities of SO- and Sp-type gauge theories explored in \cite{Hori:2011pd,Kim:2017zis} in the dual open spin chains. The first step to establishing the duality is to match the number of physical solutions in the dual models, and it is desirable to generalize the present results to spin chains with integrable open boundaries in the future. 

\paragraph{Acknowledgement}

We would like to extend our heartfelt thanks to Prof. Wenli Yang for providing us with invaluable advice and guidance. We thank Zhanqiang Bai, Kun Hao, Jue Hou, Yunfeng Jiang, Taro Kimura, J. Michael Kosterlitz, Zhangcheng Liu, Andrei S. Losev for inspiring discussions on relevant topics. We also want to thank the support from the workshop ``Tianfu Fields and Strings 2024,'' where the work is completed.
H.S. is supported by the National Natural Science Foundation of China No.12405087 and the Startup Funding of Zhengzhou University (Grant No.121-35220049). P.Z. thanks the 2024 Summer Program of the Simons Center for Geometry and Physics for kind hospitality, where part of this work was completed. R.Z. is supported by the National Natural Science Foundation of China No. 12105198 and the High-level personnel project of Jiangsu Province (JSSCBS20210709). H.~Z. is partly supported by the National Natural Science Foundation of China (Grant No.~12405083, 12475005), Shanghai Magnolia Talent Program  Pujiang Project (Grant No.~24PJA119), and was supported by the China Postdoctoral Science Foundation (Grant No.~2022M720509). H.~Z. thanks the hospitality of Xin Hua Hospital affiliated to School of Medicine, Shanghai Jiao Tong University, where part of this work was done when his arm was broken.

\bibliography{ref}

\providecommand{\href}[2]{#2}\begingroup\raggedright\begin{thebibliography}{10}

\bibitem{Kulish_2012}
P.~P. Kulish, V.~D. Lyakhovsky, and O.~V. Postnova, ``{Multiplicity functions
  for tensor powers. $A_n$-case"},''
  \href{http://dx.doi.org/10.1088/1742-6596/343/1/012070}{{\em Journal of
  Physics: Conference Series} {\bfseries 343} no.~1, (Feb, 2012) 012070}.
  \url{https://dx.doi.org/10.1088/1742-6596/343/1/012070}.

\bibitem{Curtright:2016eni}
T.~Curtright, T.~S. Van~Kortryk, and C.~K. Zachos, ``{Spin multiplicities},''
  \href{http://dx.doi.org/10.1016/j.physleta.2016.12.006}{{\em Phys. Lett. A}
  {\bfseries 381} (2017) 422--427},
  \href{http://arxiv.org/abs/1607.05849}{{\ttfamily arXiv:1607.05849
  [hep-th]}}.

\bibitem{Gyamfi_2018}
J.~A. Gyamfi and V.~Barone, ``{On the composition of an arbitrary collection of
  $SU(2)$ spins: an enumerative combinatoric approach},''
  \href{http://dx.doi.org/10.1088/1751-8121/aaa8fa}{{\em Journal of Physics A:
  Mathematical and Theoretical} {\bfseries 51} no.~10, (Feb, 2018) 105202}.
  \url{https://dx.doi.org/10.1088/1751-8121/aaa8fa}.

\bibitem{postnova2020multiplicities}
O.~Postnova and N.~Reshetikhin, ``{On multiplicities of irreducibles in large
  tensor product of representations of simple Lie algebras},'' {\em Letters in
  Mathematical Physics} {\bfseries 110} no.~1, (2020) 147--178.

\bibitem{Polychronakos:2023yhq}
A.~P. Polychronakos and K.~Sfetsos, ``{Composing arbitrarily many SU($N$)
  fundamentals},''
  \href{http://dx.doi.org/10.1016/j.nuclphysb.2023.116314}{{\em Nucl. Phys. B}
  {\bfseries 994} (2023) 116314},
  \href{http://arxiv.org/abs/2305.19345}{{\ttfamily arXiv:2305.19345
  [hep-th]}}.

\bibitem{Heitler}
W.~Heitler, ``{Zur Gruppentheorie der hom{\"o}opolaren chemischen Bindung},''
  {\em Zeitschrift f{\"u}r Physik} {\bfseries 47} no.~11, (1928) 835--858.

\bibitem{Heisenberg}
W.~Heisenberg, ``{Zur Theorie des Ferromagnetismus},'' {\em Zeitschrift f{\"u}r
  Physik} {\bfseries 49} no.~9, (1928) 619--636.

\bibitem{Bethe:1931hc}
H.~Bethe, ``{On the theory of metals. 1. Eigenvalues and eigenfunctions for the
  linear atomic chain},'' \href{http://dx.doi.org/10.1007/BF01341708}{{\em Z.
  Phys.} {\bfseries 71} (1931) 205--226}.

\bibitem{Kirillov1985}
A.~N. Kirillov, ``{Combinatorial identities, and completeness of eigenstates of
  the Heisenberg magnet},'' \href{http://dx.doi.org/10.1007/BF02105347}{{\em
  Journal of Soviet Mathematics} {\bfseries 30} no.~4, (1985) 2298--2310}.
  \url{https://doi.org/10.1007/BF02105347}.

\bibitem{Yang:1967ue}
C.~N. Yang, ``Some exact results for the many-body problem in one dimension
  with repulsive delta-function interaction,''
  \href{http://dx.doi.org/10.1103/PhysRevLett.19.1312}{{\em Physical Review
  Letters} {\bfseries 19} no.~23, (1967) 1312--1315}.
  \url{https://doi.org/10.1103/PhysRevLett.19.1312}.

\bibitem{Sutherland:1975vr}
B.~Sutherland, ``{Model for a multicomponent quantum system},''
  \href{http://dx.doi.org/10.1103/PhysRevB.12.3795}{{\em Phys. Rev. B}
  {\bfseries 12} (Nov, 1975) 3795--3805}.
  \url{https://link.aps.org/doi/10.1103/PhysRevB.12.3795}.

\bibitem{Faddeev:1996iy}
L.~D. Faddeev, ``{How algebraic Bethe ansatz works for integrable model},'' in
  {\em {Les Houches School of Physics: Astrophysical Sources of Gravitational
  Radiation}}, pp.~149--219.
\newblock 5, 1996.
\newblock \href{http://arxiv.org/abs/hep-th/9605187}{{\ttfamily
  arXiv:hep-th/9605187}}.

\bibitem{Byers:1961zz}
N.~Byers and C.~N. Yang, ``{Theoretical Considerations Concerning Quantized
  Magnetic Flux in Superconducting Cylinders},''
  \href{http://dx.doi.org/10.1103/PhysRevLett.7.46}{{\em Phys. Rev. Lett.}
  {\bfseries 7} (1961) 46--49}.

\bibitem{Bazhanov:2010ts}
V.~V. Bazhanov, T.~Lukowski, C.~Meneghelli, and M.~Staudacher, ``{A Shortcut to
  the Q-Operator},''
  \href{http://dx.doi.org/10.1088/1742-5468/2010/11/P11002}{{\em J. Stat.
  Mech.} {\bfseries 1011} (2010) P11002},
  \href{http://arxiv.org/abs/1005.3261}{{\ttfamily arXiv:1005.3261 [hep-th]}}.

\bibitem{Hou:2023ndn}
J.~Hou, Y.~Jiang, and R.-D. Zhu, ``{Spin-$s$ rational $Q$-system},''
  \href{http://dx.doi.org/10.21468/SciPostPhys.16.4.113}{{\em SciPost Phys.}
  {\bfseries 16} no.~4, (2024) 113},
  \href{http://arxiv.org/abs/2303.07640}{{\ttfamily arXiv:2303.07640
  [hep-th]}}.

\bibitem{Kirillov_1987}
A.~N. Kirillov, ``{Completeness of states of the generalized Heisenberg
  magnet},'' \href{http://dx.doi.org/10.1007/bf01104977}{{\em Journal of Soviet
  Mathematics} {\bfseries 36} no.~1, (Jan, 1987) 115--128}.
  \url{https://doi.org/10.1007%2Fbf01104977}.

\bibitem{Nekrasov:2009uh}
N.~A. Nekrasov and S.~L. Shatashvili, ``{Supersymmetric vacua and Bethe
  ansatz},'' \href{http://dx.doi.org/10.1016/j.nuclphysbps.2009.07.047}{{\em
  Nucl. Phys. B Proc. Suppl.} {\bfseries 192-193} (2009) 91--112},
  \href{http://arxiv.org/abs/0901.4744}{{\ttfamily arXiv:0901.4744 [hep-th]}}.

\bibitem{Nekrasov:2009ui}
N.~A. Nekrasov and S.~L. Shatashvili, ``{Quantum integrability and
  supersymmetric vacua},'' \href{http://dx.doi.org/10.1143/PTPS.177.105}{{\em
  Prog.Theor.Phys.Suppl.} {\bfseries 177} (2009) 105--119},
  \href{http://arxiv.org/abs/0901.4748}{{\ttfamily arXiv:0901.4748 [hep-th]}}.
21 pp., short version II, conference in honour of T.Eguchi's 60th anniversary.

\bibitem{Shu:2022vpk}
H.~Shu, P.~Zhao, R.-D. Zhu, and H.~Zou, ``{Bethe-state counting and Witten
  index},'' \href{http://dx.doi.org/10.21468/SciPostPhys.15.3.103}{{\em SciPost
  Phys.} {\bfseries 15} no.~3, (2023) 103},
  \href{http://arxiv.org/abs/2210.07116}{{\ttfamily arXiv:2210.07116
  [hep-th]}}.

\bibitem{Witten:1993yc}
E.~Witten, ``{Phases of $N=2$ theories in two-dimensions},''
  \href{http://dx.doi.org/10.1016/0550-3213(93)90033-L}{{\em Nucl. Phys. B}
  {\bfseries 403} (1993) 159--222},
  \href{http://arxiv.org/abs/hep-th/9301042}{{\ttfamily arXiv:hep-th/9301042}}.

\bibitem{Nepomechie:2013mua}
R.~I. Nepomechie and C.~Wang, ``{Algebraic Bethe ansatz for singular
  solutions},'' \href{http://dx.doi.org/10.1088/1751-8113/46/32/325002}{{\em J.
  Phys. A} {\bfseries 46} (2013) 325002},
  \href{http://arxiv.org/abs/1304.7978}{{\ttfamily arXiv:1304.7978 [hep-th]}}.

\bibitem{Nepomechie:2014hma}
R.~I. Nepomechie and C.~Wang, ``{Twisting singular solutions of Bethe's
  equations},'' \href{http://dx.doi.org/10.1088/1751-8113/47/50/505004}{{\em J.
  Phys. A} {\bfseries 47} no.~50, (2014) 505004},
  \href{http://arxiv.org/abs/1409.7382}{{\ttfamily arXiv:1409.7382 [math-ph]}}.

\bibitem{Kulish:1979cr}
P.~P. Kulish and N.~Y. Reshetikhin, ``{Generalized Heisenberg ferromagnet and
  the Gross-Neveu model},'' {\em Sov. Phys. JETP} {\bfseries 53} (1981)
  108--114.

\bibitem{Kulish:1983rd}
P.~P. Kulish and N.~Y. Reshetikhin, ``{Diagonalisation of GL($N$) invariant
  transfer matrices and quantum $N$-wave system (Lee model)},''
  \href{http://dx.doi.org/10.1088/0305-4470/16/16/001}{{\em J. Phys. A}
  {\bfseries 16} (1983) L591--L596}.

\bibitem{Slavnov:2019hdn}
N.~A. Slavnov, ``{Introduction to the nested algebraic Bethe ansatz},''
  \href{http://dx.doi.org/10.21468/SciPostPhysLectNotes.19}{{\em SciPost Phys.
  Lect. Notes} {\bfseries 19} (2020) 1},
  \href{http://arxiv.org/abs/1911.12811}{{\ttfamily arXiv:1911.12811
  [math-ph]}}.

\bibitem{ACDFR05}
D.~{Arnaudon}, N.~{Cramp{\'e}}, A.~{Doikou}, L.~{Frappat}, and
  {\'E}.~{Ragoucy}, ``{Analytical Bethe ansatz for closed and open
  gl($\mathcal{N}$)-spin chains in any representation},''
  \href{http://dx.doi.org/10.1088/1742-5468/2005/02/P02007}{{\em Journal of
  Statistical Mechanics: Theory and Experiment} {\bfseries 2005} no.~2, (Feb.,
  2005) 02007}, \href{http://arxiv.org/abs/math-ph/0411021}{{\ttfamily
  arXiv:math-ph/0411021 [math-ph]}}.

\bibitem{Freund:1956ro}
J.~E. Freund and A.~N. Pozner, ``{Some Results on Restricted Occupancy
  Theory},'' \href{http://dx.doi.org/10.1214/aoms/1177728278}{{\em The Annals
  of Mathematical Statistics} {\bfseries 27} no.~2, (1956) 537 -- 540}.
  \url{https://doi.org/10.1214/aoms/1177728278}.

\bibitem{Ramond_2010}
P.~Ramond, {\em Group Theory: A Physicist’s Survey}.
\newblock Cambridge University Press, 2010.

\bibitem{SZZZ}
H.~Shu, P.~Zhao, R.-D. Zhu, and H.~Zou {\em In Progress} .

\bibitem{fulton1991representation}
W.~Fulton and J.~Harris, {\em Representation Theory: A First Course}.
\newblock Graduate Texts in Mathematics. Springer New York, 1991.
\newblock \url{https://books.google.co.jp/books?id=6GUH8ARxhp8C}.

\bibitem{PhysRevB.25.5935}
K.~Furuya and J.~H. Lowenstein, ``{Bethe-ansatz approach to the Kondo model
  with arbitrary impurity spin},''
  \href{http://dx.doi.org/10.1103/PhysRevB.25.5935}{{\em Phys. Rev. B}
  {\bfseries 25} (May, 1982) 5935--5952}.
  \url{https://link.aps.org/doi/10.1103/PhysRevB.25.5935}.

\bibitem{Dafnis:2007}
S.~D. Dafnis, F.~S. Makri, and A.~N. Philippou, ``{Restricted occupancy of $s$
  kinds of cells and generalized Pascal triangles},'' {\em Fibonacci Quart.}
  {\bfseries 45} (2007) 347--356.

\bibitem{kac1977characters}
V.~Kac, ``{Characters of typical representations of classical Lie
  superalgebras},'' {\em Communications in Algebra} {\bfseries 5} no.~8, (1977)
  889--897.

\bibitem{frappat:hal-00376660}
L.~Frappat, P.~Sorba, and A.~Sciarrino, {\em {Dictionary on Lie algebras and
  superalgebras}}.
\newblock {Academic Press (London)}, 2000.
\newblock \url{https://hal.science/hal-00376660}.

\bibitem{Saleur:1989nw}
H.~Saleur, ``{Symmetries of the XX chain and applications},'' in {\em {Trieste
  Conference on Recent Developments in Conformal Field Theories}}.
\newblock 1989.

\bibitem{Bars:1982se}
I.~Bars, B.~Morel, and H.~Ruegg, ``{Kac-dynkin Diagrams and Supertableaux},''
  \href{http://dx.doi.org/10.1063/1.525970}{{\em J. Math. Phys.} {\bfseries 24}
  (1983) 2253}.

\bibitem{Sarkar_1990}
S.~Sarkar, ``{Bethe-ansatz solution of the t-J model},''
  \href{http://dx.doi.org/10.1088/0305-4470/23/9/002}{{\em Journal of Physics
  A: Mathematical and General} {\bfseries 23} no.~9, (May, 1990) L409--L414}.
  \url{https://doi.org/10.1088/0305-4470/23/9/002}.

\bibitem{Chao_1977}
K.~A. Chao, J.~Spalek, and A.~M. Oles, ``Kinetic exchange interaction in a
  narrow s-band,'' \href{http://dx.doi.org/10.1088/0022-3719/10/10/002}{{\em
  Journal of Physics C: Solid State Physics} {\bfseries 10} no.~10, (May, 1977)
  L271--L276}. \url{https://doi.org/10.1088%2F0022-3719%2F10%2F10%2F002}.

\bibitem{Foerster:1992ud}
A.~Foerster and M.~Karowski, ``{Completeness of the Bethe states for the
  supersymmetric t-J model},''
  \href{http://dx.doi.org/10.1103/PhysRevB.46.9234}{{\em Phys. Rev. B}
  {\bfseries 46} (Oct, 1992) 9234--9236}.
  \url{https://link.aps.org/doi/10.1103/PhysRevB.46.9234}.

\bibitem{FOERSTER1993611}
A.~Foerster and M.~Karowski, ``{Algebraic properties of the Bethe ansatz for an
  $spl(2,1)$-supersymmetric t-J model},''
  \href{http://dx.doi.org/https://doi.org/10.1016/0550-3213(93)90665-C}{{\em
  Nuclear Physics B} {\bfseries 396} no.~2, (1993) 611--638}.
  \url{https://www.sciencedirect.com/science/article/pii/055032139390665C}.

\bibitem{Hao:2013rza}
W.~Hao, R.~I. Nepomechie, and A.~J. Sommese, ``{Singular solutions, repeated
  roots and completeness for higher-spin chains},''
  \href{http://dx.doi.org/10.1088/1742-5468/2014/03/P03024}{{\em J. Stat.
  Mech.} {\bfseries 1403} (2014) P03024},
  \href{http://arxiv.org/abs/1312.2982}{{\ttfamily arXiv:1312.2982 [math-ph]}}.

\bibitem{Hao:2013jqa}
W.~Hao, R.~I. Nepomechie, and A.~J. Sommese, ``{Completeness of solutions of
  Bethe's equations},''
  \href{http://dx.doi.org/10.1103/PhysRevE.88.052113}{{\em Phys. Rev. E}
  {\bfseries 88} no.~5, (2013) 052113},
  \href{http://arxiv.org/abs/1308.4645}{{\ttfamily arXiv:1308.4645 [math-ph]}}.

\bibitem{Marboe:2016yyn}
C.~Marboe and D.~Volin, ``{Fast analytic solver of rational Bethe equations},''
  \href{http://dx.doi.org/10.1088/1751-8121/aa6b88}{{\em J. Phys. A} {\bfseries
  50} no.~20, (2017) 204002}, \href{http://arxiv.org/abs/1608.06504}{{\ttfamily
  arXiv:1608.06504 [math-ph]}}.

\bibitem{Krichever:1996qd}
I.~Krichever, O.~Lipan, P.~Wiegmann, and A.~Zabrodin, ``{Quantum integrable
  systems and elliptic solutions of classical discrete nonlinear equations},''
  \href{http://dx.doi.org/10.1007/s002200050165}{{\em Commun. Math. Phys.}
  {\bfseries 188} (1997) 267--304},
  \href{http://arxiv.org/abs/hep-th/9604080}{{\ttfamily arXiv:hep-th/9604080}}.

\bibitem{Kazakov:2015efa}
V.~Kazakov, S.~Leurent, and D.~Volin, ``{T-system on T-hook: Grassmannian
  Solution and Twisted Quantum Spectral Curve},''
  \href{http://dx.doi.org/10.1007/JHEP12(2016)044}{{\em JHEP} {\bfseries 12}
  (2016) 044}, \href{http://arxiv.org/abs/1510.02100}{{\ttfamily
  arXiv:1510.02100 [hep-th]}}.

\bibitem{Sklyanin:1988yz}
E.~K. Sklyanin, ``{Boundary Conditions for Integrable Quantum Systems},''
  \href{http://dx.doi.org/10.1088/0305-4470/21/10/015}{{\em J. Phys. A}
  {\bfseries 21} (1988) 2375--2389}.

\bibitem{Cao:2013nza}
J.~Cao, W.~Yang, K.~Shi, and Y.~Wang, ``{Off-diagonal Bethe ansatz and exact
  solution of a topological spin ring},''
  \href{http://dx.doi.org/10.1103/PhysRevLett.111.137201}{{\em Phys. Rev.
  Lett.} {\bfseries 111} no.~13, (2013) 137201},
  \href{http://arxiv.org/abs/1305.7328}{{\ttfamily arXiv:1305.7328
  [cond-mat.stat-mech]}}.

\bibitem{Kimura:2020bed}
T.~Kimura and R.-D. Zhu, ``{Bethe/Gauge Correspondence for SO/Sp Gauge Theories
  and Open Spin Chains},''
  \href{http://dx.doi.org/10.1007/JHEP03(2021)227}{{\em JHEP} {\bfseries 03}
  (2021) 227}, \href{http://arxiv.org/abs/2012.14197}{{\ttfamily
  arXiv:2012.14197 [hep-th]}}.

\bibitem{Ding:2023auy}
X.-M. Ding and T.~Zhang, ``{Bethe/Gauge correspondence for ABCDEFG-type 3d
  gauge theories},'' \href{http://dx.doi.org/10.1007/JHEP04(2023)036}{{\em
  JHEP} {\bfseries 04} (2023) 036},
  \href{http://arxiv.org/abs/2303.03102}{{\ttfamily arXiv:2303.03102
  [hep-th]}}. [Erratum: JHEP 06, 177 (2023)].

\bibitem{Ding:2023lsk}
X.-M. Ding and T.~Zhang, ``{Bethe/Gauge correspondence for linear quiver
  theories with ABCD gauge symmetry and spin chains},''
  \href{http://dx.doi.org/10.1016/j.nuclphysb.2023.116222}{{\em Nucl. Phys. B}
  {\bfseries 991} (2023) 116222},
  \href{http://arxiv.org/abs/2303.04575}{{\ttfamily arXiv:2303.04575
  [hep-th]}}.

\bibitem{Ding:2023nkv}
X.-M. Ding and T.~Zhang, ``{Langlands Dualities through Bethe/Gauge
  Correspondence for 3d Gauge Theories},''
  \href{http://arxiv.org/abs/2312.13080}{{\ttfamily arXiv:2312.13080
  [math-ph]}}.

\bibitem{Wang:2024zcr}
Z.~Wang and R.-D. Zhu, ``{Bethe/Gauge correspondence for $A_{N}$ spin chains
  with integrable boundaries},''
  \href{http://dx.doi.org/10.1007/JHEP04(2024)112}{{\em JHEP} {\bfseries 04}
  (2024) 112}, \href{http://arxiv.org/abs/2401.00764}{{\ttfamily
  arXiv:2401.00764 [hep-th]}}.

\bibitem{Hori:2011pd}
K.~Hori, ``{Duality In Two-Dimensional $\mathcal{N} = (2,2)$ Supersymmetric
  Non-Abelian Gauge Theories},''
  \href{http://dx.doi.org/10.1007/JHEP10(2013)121}{{\em JHEP} {\bfseries 10}
  (2013) 121}, \href{http://arxiv.org/abs/1104.2853}{{\ttfamily arXiv:1104.2853
  [hep-th]}}.

\bibitem{Kim:2017zis}
H.~Kim, S.~Kim, and J.~Park, ``{2D Seiberg-like dualities for orthogonal gauge
  groups},'' \href{http://dx.doi.org/10.1007/JHEP10(2019)079}{{\em JHEP}
  {\bfseries 10} (2019) 079}, \href{http://arxiv.org/abs/1710.06069}{{\ttfamily
  arXiv:1710.06069 [hep-th]}}.

\end{thebibliography}\endgroup

\end{document}